\documentclass[11pt,nofootinbib,superscriptaddress]{revtex4-1}
\usepackage[english]{babel}
\usepackage[T1]{fontenc}
\usepackage[latin9]{inputenc}
\usepackage{geometry}
\usepackage{amsmath}
\usepackage{amssymb}
\usepackage{graphicx}
\usepackage{color}
\usepackage{esint}

\makeatletter
\@addtoreset{figure}{section}

\makeatother

\newcommand{\figref}[1]{Fig. \ref{#1}}

\newcommand{\makeSymbol}[1]{\mathord{\vcenter{\hbox{#1}}}}

\def\beq{\begin{equation}}
\def\eeq{\end{equation}}

\usepackage{amsthm}

\newtheorem*{lemma*}{Lemma}
\newtheorem*{theorem*}{Theorem}

\def\tip{\makeSymbol{\includegraphics{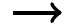}}}
\def\feather{\makeSymbol{\includegraphics{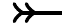}}}
\newcommand{\vgraph}{\mathfrak{n}}
\def\tsum{\mathop{\tilde{\sum}}}

\newcommand{\N}{\mathbb{N}}

\newcommand{\SU}{\mathrm{SU}}
\newcommand{\su}{\mathfrak{su}(2)}

\renewcommand{\d}{\mathrm{d}}

\newcommand{\Tr}{\mathrm{Tr}}

\newcommand{\ket}[1]{|#1\rangle}
\newcommand{\bra}[1]{\langle #1|}
\newcommand{\scal}[1]{\langle #1\rangle}

\newcommand{\dtensor}[3]{\begin{pmatrix}  \multicolumn{2}{c}{#1}\\#2&#3\end{pmatrix}}
\newcommand{\threej}[6]{\begin{pmatrix}  #1&#3&#5\\#2&#4&#6\end{pmatrix}}
\newcommand{\sixj}[6]{\left\{ \begin{array}{ccc}  #1&#3&#5\\#2&#4&#6\end{array}\right\}}
\newcommand{\tinysixj}[6]{\{\begin{smallmatrix}  #1&#3&#5\\#2&#4&#6\end{smallmatrix}\}}
\newcommand{\ninej}[9]{\left\{ \begin{array}{ccc}  #1&#2&#3\\#4&#5&#6\\#7&#8&#9\end{array}\right\}}
\newcommand{\tinyninej}[9]{\left\{\begin{smallmatrix}  #1&#2&#3\\#4&#5&#6\\#7&#8&#9\end{smallmatrix}\right\}}

 \date{\normalsize \today} 

\begin{document}
\author{Emanuele Alesci}
\email{Emanuele.Alesci@fuw.edu.pl}
\affiliation{Instytut Fizyki Teoretycznej, Uniwersytet Warszawski, ul. Ho{\.z}a 69, 00-681 Warszawa, Poland, EU}

\author{Klaus Liegener}
\email{Klaus.Liegener@gravity.fau.de}
\affiliation{Universit\"at Erlangen, Institut f\"ur Theoretische Physik III, Lehrstuhl f\"ur Quantengravitation\\ Staudtstrasse 7, D-91058 Erlangen, EU}

\author{Antonia Zipfel}
\email{antonia.zipfel@gravity.fau.de}
\affiliation{Universit\"at Erlangen, Institut f\"ur Theoretische Physik III, Lehrstuhl f\"ur Quantengravitation\\ Staudtstrasse 7, D-91058 Erlangen, EU}

\title{\bf Matrix Elements of Lorentzian Hamiltonian Constraint in LQG}

\begin{abstract}
\begin{center}
{\bf Abstract}

The Hamiltonian constraint is the key element of the canonical formulation of LQG coding its dynamics. In Ashtekar-Barbero variables it naturally splits into the so called Euclidean and Lorentzian parts. However, due to the high complexity of this operator, only the matrix elements of the Euclidean part have been considered so far. Here we evaluate the action of the full constraint, including the Lorentzian part. The computation requires an heavy use of $SU(2)$ recoupling theory and several tricky identities among n-j symbols are used to find the final result: these identities, together with the graphical calculus used to derive them, also simplify the Euclidean constraint and are of general interest in LQG computations.
\end{center}
\end{abstract}
\maketitle

\section{Introduction}

In the Hamiltonian formulation General Relativity (GR) is completely governed by the Diffeomorphism and Hamiltonian constraints. For many years the complicated structure of these constraints prevented a quantization of the theory, until Ashtekar \cite{Ashtekar:1986yd} suggested to replace the "old" metric variables by connections and tetrads. Indeed in this variables GR resembles other gauge theories, like Yang-Mills theory, whereupon one encounters an $\SU(2)$-Gauss constraint in addition to Diffeomorphism and Hamiltonian constraints. This formulation was further improved by Barbero \cite{barbero} and serves today as the classical starting point of Loop Quantum Gravity (LQG) \cite{lqgcan1,lqgcan2,lqgcan3}.

LQG follows the Dirac quantization program \cite{DiracQM1964} for constrained systems, i.e. one introduces a preliminary kinematical Hilbert space on which the constraints can be represented by operators and then seeks for the kernel of these operators defining the physical Hilbert space. 
The Gauss constraint is solved by introducing a so called "spin network" basis \cite{Rovelli:1995ac}, naturally leading to a combinatorial discrete structure of space-time similar to the one proposed by Penrose \cite{Penrose}, while the Diffeomorphisms constraint is solved by considering equivalence class of spin networks under diffeomorphisms denoted "s-knots" \cite{Ashtekar:1995zh}.

A major obstacle for completing the canonical quantization program in LQG is the implementation of the Hamiltonian constraint $S$. The difficulties are mainly caused by the non-polynomial structure of $S$ and the weight factor $1/\sqrt{\det(q)}$ determined by the intrinsic metric $q:=q_{ab}$ on the initial hypersurface $\Sigma$. In fact, Ashtekar \cite{Ashtekar:1986yd} was motivated by the observation that the Hamiltonian constraint can be casted into a polynomial form when the metric variables are replaced by triads and \emph{complex} connections. Even though this simplifies the constraint one has to deal instead with difficult reality conditions. This reality structure is of course trivial when the theory is formulated in \emph{real} connection variables as suggested by Barbero \cite{barbero}. Unfortunately, the constraint remains to be non-polynomial in these variables. It was then proposed to absorb the weight $1/\sqrt{\det(q)}$ in the Lapse-function. But it turned out \cite{Thiemann:1996ay} that this density weight is crucial in order to obtain a finite, background independent operator.

After many efforts \cite{Early-hamiltonians}, Thiemann  \cite{Thiemann96a} discovered that both problems, the non-polynomiality and the appearance of the weight factor, can be solved by expressing the inverse triads through the Poisson bracket of volume and connections  ("Thiemann trick"). This trick made it possible to construct a finite, anomaly-free operator that corresponds to the non-rescaled Hamiltonian constraint \cite{Thiemann96a, Thiemann96b, lqgcan2} and acts by changing the underlying graph of the spin networks. The formal solution \cite{Thiemann96b} to this constraint are superpositions of s-knot states with "dressed nodes", that are nodes with a spider-web like structure. Criticism appeared \cite{dubbi} mostly concerning the "ultralocal" character of the construction and regularization ambiguities. However until now, this is the only known scheme able to realize an anomaly free quantization of the Dirac algebra at least on shell.

This construction is at the heart of many other approaches within canonical LQG, as the master constraint program \cite{master}, Algebraic Quantum Gravity (AQG) \cite{Giesel:2006I}, most recent models with matter \cite{Domagala, Giesel:2012, Husain:2011tk} and symmetry reduced models like Loop Quantum Cosmology \cite{Ashtekar:2011ni,  Bojowald:2008zzb}. Also the covariant approach (spin foam models) \cite{lqgcov} 
is motivated by the idea of realizing the "time-evolution" generated by a graph-changing Hamiltonian \cite{ReisenbergerRovelli97}. In fact, it is hoped that the spin-foam model might provide a physical scalar product for canonical LQG (see e.g. \cite{SFscalarproduct}). The attempt to match both approaches (not only heuristically) has led to new regularization schemes for the Hamiltonian constraint \cite{iopro, ioreg} and also to the discovery of new physical states in the canonical model \cite{ioantonia}.

Despite it's central role for LQG the action of $S$ has been analyzed explicitly in only very few examples \cite{Borissovetal97, Gaul, ioantonia} and these are confined to the Euclidean part of the constraint only. This is mainly due to two reasons: first the presence of the volume operator \cite{RovelliSmolin95, AshtekarLewand98} and second the non trivial recoupling of $\SU(2)$ irreducible representations. The volume operator in LQG has been studied intensively in \cite{Thiemann:1996au, DePietriRovelli96, Brunnemann:2004xi}. Yet, the matrix elements are getting very complicated the more edges are involved so that one has to apply numerical methods \cite{Brunnemann:2007ca, Brunnemann:2007ca2} to evaluate it. The second difficulty appears due to the regularization of the connections by holonomies. The corresponding operators act by multiplication which produces several Clebsh-Gordan decompositions and modifications to the intertwiners between $\SU(2)$ representations at the nodes. 

In this paper we explicitly compute the matrix elements of the \emph{full} Lorentzian constraint in the Thiemann prescription for trivalent nodes. The final result still depends on the matrix elements of the volume which are unknown in closed form, but  in principle computable. In the course of the evaluation several recoupling identities will be proven, which greatly simplify the final result and are expected to be useful in all the computations involving curvature loops or the "Thiemann trick". 
The resulting compact formula presented here opens the possibility to test the implementation of the constraint by simulations, analyze the behavior of $S$ in a large j-limit or further develop the methods of \cite{ioantonia}.

The article is organized as follows: In Section \ref{HConstraint} Thiemann's construction for the Euclidean and Lorentzian term is  briefly reviewed and in Section \ref{tools} the main recoupling identities are introduced that will then be applied in Section \ref{sec:main1} to the Euclidean constraint leading to a new and very compact expression for it. Finally in Section \ref{sec:main2}, we present the matrix elements of the Lorentzian part. Section \ref{conclusion} is left for concluding remarks and an Appendix with further details on n-$j$ symbols and the Volume operator is included to make the manuscript self-contained.

\section{Hamiltonian Constraint} 
\label{HConstraint}

\subsection{Classical constraint}
Let $e_a^i$ be a triad on a smooth, spatial hypersurface $\Sigma$ defining the intrinsic metric $q_{ab}=\delta_{ij} e^i_a e^j_b$. Here, $a,b=1,2,3$ are tensorial and $i,j=1,2,3$ are $\su$-indices. In the following, $\Gamma^i_a$ denotes the spin-connection associated to $e^i_a$ and $K_{ab}$ the extrinsic curvature. Given $K^i_a:=\mathrm{sgn}(\det(e^j_c)) e^b_i K_{ab}$, it can be shown that the densitized inverse triad $E^a_i:=\sqrt{\det(q)} e^a_i$ and $A^i_a:=\Gamma^i_a+\gamma K^i_a$ form a canonical conjugated pair, 
\beq
\{A^i_a(x),E^b_j(y)\}=\kappa \delta^i_j \delta^b_a \delta^{(3)}(x,y)~,
\eeq
where $\gamma$ is a real non-zero parameter and $\kappa= \frac{8 \pi G}{c^3}\gamma$. If $F_{ab}$ denotes the curvature of $A$ and $s$ the signature of the space-time metric then the classical Hamiltonian constraint is of the form
\beq
S=\frac{1}{\sqrt{\det(q)}}{\rm Tr}\{(F_{ab}-(\gamma^2-s)[K_a, K_b])[E^a,E^b]\}~.
\label{constraint}
\eeq

According to \cite{Thiemann96a}, the square root in \eqref{constraint} can be absorbed by using Poisson brackets of the connection $A$ with the Volume,
\beq
V:=V(\Sigma)=\int_{\Sigma}\d x^3\sqrt{\det(q)}~,
\eeq 
and the integrated curvature,  
\beq
K:=\int_{\Sigma}\d x^3 E^a_j K^j_a~.
\eeq
More explicitly, inserting
\begin{gather}
\begin{gathered}
  \label{keyID1}
\frac{[E^a,E^b]^i}{\sqrt{\det(q)}} =-\frac{2}{\kappa}\epsilon^{abc} \{A_c^i(x),V\}\; \text{ and }\;
K^i_a=\frac{1}{\kappa}\{A^i_a(x),K\}
\end{gathered}
\end{gather}
in \eqref{constraint} yields
\beq
  H= \frac{2}{\kappa} \epsilon^{abc}\,{\rm Tr}(F_{ab}  \{ A_c, V \})             
  \label{H_E_N}  
 \eeq
 and
 \beq
 \label{eqn:T}
  T:=S-H
     = (s-\gamma^2)\frac{4}{\kappa^3}   \epsilon^{abc}\, {\rm Tr}(\{A_a,K\} \{A_b,K\} \{ A_c, V \})~,
\eeq
where $S$ was split into an Euclidean part $H$ and remaining constraint $T$, which vanishes for $\gamma^2=1$ and $s=1$. The second part \eqref{eqn:T} can be further modified by expressing $K$ through the 'time' derivative of the volume:
\beq
K=-\frac{1}{\kappa \gamma}\{V,\int_{\Sigma}\d x^3 H\}~.
\eeq
Thus $S$ is completely determined by the connection $A$, the curvature $F$ and the volume $V$ all of which have well-defined operator analogous in LQG. 

\subsection{Quantization}

Prior to quantization the local expressions \eqref{H_E_N} and \eqref{eqn:T} must be smeared and regularized. That is to say the connection and curvature of the integrated constraint $S[N]:=\int_{\Sigma}N(x) S(x)$ are replaced by holonomies along edges and loops respectively. Of course the properties of the operator depend highly on the chosen regularization and up to now there are several different models on the market (see e.g. \cite{Thiemann96a,lqgcan1, ioreg}). Here, we follow the original proposal \cite{Thiemann96a} because it is comparatively easy and leads to anomaly free and finite operators. \\
Since the Euclidean constraint $H[N]$ depends linearly on the volume and the volume operator is acting locally on the nodes it suffices to construct a regularization in the neighborhood of a node $\vgraph$ in a given graph $\Gamma$ and then extend it to all of $\Sigma$. Let $s_I$ be a segment of an edge $e_I$ incident at $\vgraph$ and $\alpha_{IJ}$ the loop generated by $s_I$ and $s_J$. I.e. $\alpha_{IJ}=s_I \circ a_{IJ} \circ s_J^{-1}$ where $a_{IJ}$ is a semi-analytic arc which only intersects with $\Gamma$ in the endpoints of $s_I$ and $s_J$ (see Fig. \ref{tetrahedron}). In this way, any three (non-planar) edges $e_I,e_J$ and $e_K$ incident at $\vgraph$ constitute an \emph{elementary tetrahedron} $\Delta_{IJK}$. Starting with $ \Delta_{IJK}$ one can now construct seven additional tetrahedra (see \cite{Thiemann96a} for details), such that the eight tetrahedra including $\Delta_{IJK}$ cover a neighborhood of $\vgraph$. Afterwards this is extended to a full triangulation\footnote{Note, for each triple of edges adjacent to a node one constructs an \emph{independent} triangulation.} $\mathfrak{T}(I,J,K)$ of $\Sigma$ and $H[N]$ is decomposed into $\sum\limits_{\Delta \in \mathfrak{T}} H_{\Delta}[N]$. \\
\begin{figure}
  \begin{center}
\includegraphics{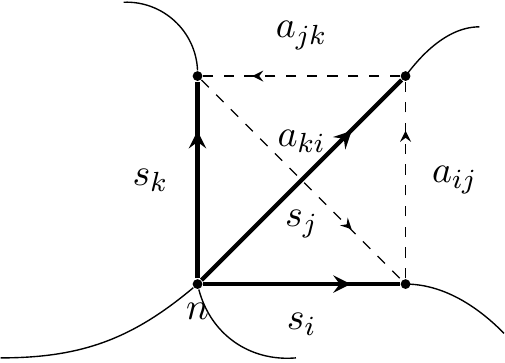}\\
   \parbox{8 cm}{\caption[]{\label{tetrahedron} \small
       An elementary tetrahedron $\Delta \in T$ constructed
      by adapting it to a graph $\Gamma$.}}
  \end{center}
\end{figure}
On the elementary tetrahedron $\Delta_{IJK}$ the connection $A$ and the curvature $F$ are regularized as usual by smearing along $s_K$ and $\alpha_{IJ}$ respectively so that the regularized constraint is defined by
\begin{align}
\begin{split}
      H_{\Delta_{IJK}}[N]  :&=\sum_{i,j,k\in \{I,J,K\}}  \frac{2}{3\kappa} \, N(\vgraph) \, \epsilon^{ijk} \, \mathrm{Tr}\Big[h_{\alpha_{ij}} 
     	h_{s_{k}} \big\{h^{-1}_{s_{k}},V\big\}\Big]~,
 \label{H_delta}
\end{split}
\end{align}
where $h_s$ is the holonomy alongside $s$ and $N(\vgraph)$ is the value of the lapse function $N(x)$ at $\vgraph$. In the article \cite{Gaul} it was proposed to generalize this by considering holonomies in an arbitrary irrep $m$ that yields
 \beq
   \label{Hm_delta:classical2}
   H^m_{\Delta_{IJK}}[N]:=\sum_{i,j,k\in \{I,J,K\}} \frac{N(\vgraph)}{N^2_m\kappa} \,  \, \epsilon^{ijk} \,
   \mathrm{Tr}\Big[h^{(m)}_{\alpha_{ij}} h^{(m)}_{s_{k}} \big\{h^{(m)}_{s_{k}^{-1}},V\big\}\Big] ~,
\eeq
with the normalization factor $N_m^2=(2m+1)m(m+1)$. At this point, the quantization of \eqref{Hm_delta:classical2} is straight forward. Its action on a cylindrical function $T_s$ on a spin net $s$ with underlying graph $\Gamma$ is 
\beq
\label{eqn:quantumH}
\hat{H}[N] \,T_s=\frac{i}{N^2_m\,\kappa\, l^2_p}\sum_{\vgraph\in\Gamma}\frac{8 N(\vgraph)}{E(\vgraph)}\sum_{\vgraph(\Delta)=\vgraph} \epsilon^{ijk} \,\Tr\left( h^{(m)}_{\alpha_{ij}} h^{(m)}_{s_k}[h^{(m)}_{s_k^{-1}},\hat{V}]\right)T_s~.
\eeq
The first sum is running over all nodes of $\Gamma$ and the second over all elementary tetrahedra. $l_p$ is the Planck length. Because every vertex is surrounded by 8 tetrahedra the smearing of $H$ in the neighborhood of $\vgraph$ yields eight times the same factor. Apart from that, at an $m$-valent vertex $\vgraph$ there are $E(\vgraph):=\binom{m}{3}$ elementary tetrahedra each of which determines an adapted triangulation. Therefore we need to divide by $E(\vgraph)$ to avoid over counting.\\
A huge advantage of the operator defined above is that it is anomaly free, i.e. that the commutator of two constraints $\hat{H}$ is vanishing up to diffeomorphisms. This is mainly due to the behavior of the volume which is vanishing on coplanar nodes\footnote{This is only true for the version defined by Ashtekar and Lewandowski \cite{AshtekarLewand98}, not the one introduced by Smolin and Rovelli \cite{RovelliSmolin95}. }. But, in fact, $H$ only generates such nodes (see \figref{principal action}) so that a second Hamiltonian acts again only on the 'old' ones. \\
\begin{figure}
\centering
	\includegraphics{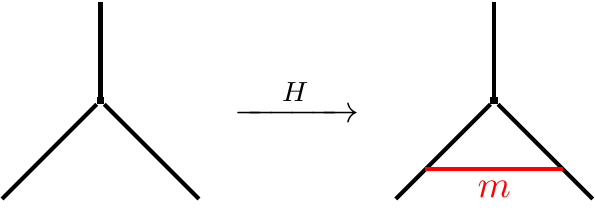}
\hspace*{40pt}
	\includegraphics{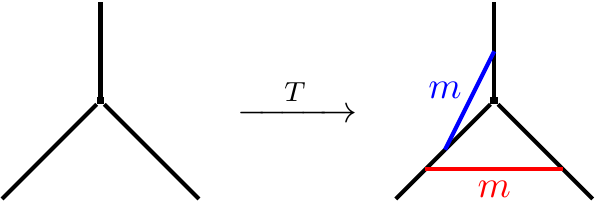}\
   \parbox{15cm}{\caption[]{\label{principal action} \small
The Euclidean constraint $H$ creates a new link (red) joining two edges at the same vertex due to the regularization of the curvature. Note, the new nodes are coplanar whereby two of the adjacent edges are even colinear. On the right hand side the principal action of $T$ is visualized: The red link is created by the first extrinsic curvature while the blue one is created by the second curvature operator that is regularized along a tetrahedron lying inside of the first one.}}
\end{figure}
The remaining part of the constraint $T[N]$ can be quantized similarly. However, the regularization is a bit more involved because the extrinsic curvature must be regularized separately. In principle $T$ adds two new links each of which is created by one operator $\hat{K}:=\frac{i}{l_p^2\,\gamma}[\hat{V},\hat{H}[1]]$. To insure that the full constraint is still anomaly free only coplanar nodes should be generated. That means, it must never happen  that the two new links have a common intersection. Because $\hat{H}$ is acting locally the second extrinsic curvature should be therefore regulated along tetrahedra lying inside of the first ones (see \figref{principal action}). The holonomies in $T$ can be regulated as above such that finally
\begin{gather}
\label{eqn:quantumT}
\hat{T}[N] \,T_s=\frac{2 i (s-\gamma^2) }{N^2_m\, \kappa\, l^6_p}\sum_{\vgraph\in\Gamma^{(0)}}\frac{8 N(\vgraph)}{E(\vgraph)}
\sum_{\vgraph(\Delta)=\vgraph} \epsilon^{ijk} \,\Tr\left( h^{(m)}_{s_i}[h^{(m)}_{s_i^{-1}},\hat{K}]h^{(m)}_{s_j}[h^{(m)}_{s_j^{-1}},\hat{K}] h^{(m)}_{s_k}[h^{(m)}_{s_k^{-1}},\hat{V}]\right)T_s~.
\end{gather}
In \eqref{eqn:quantumH} and \eqref{eqn:quantumT} the triangulation $\mathfrak{T}$ serves as a regulator. This regularization dependence can be removed by in a suitable operator topology \footnote{
On diffeo\-mor\-phism invariant states $\phi\in\mathcal{H}_{diff}\subset\mathcal{H}^{\ast}_{kin}$ the regulator dependence drops out trivially because two operators $\hat{S}$ and $\hat{S}'$ that are related by a refinement of the triangulation differ only in the size of the loops. This implies 
$
\langle \phi ,\hat{S}\psi \rangle = \langle \phi, \hat{S}' \psi \rangle~,
$
(see \cite{Thiemann96a} for details).}.

\section{Computational tools}
\label{tools}
In this section the tools for computing the matrix elements of the Hamiltonian constraint are introduced and some identities, which are important for the latter, are proven. 

\subsection{Graphical calculus}

The evaluation of the above constraint is mainly based on recoupling theory of $\SU(2)$. In this context it has proven beneficial to work with graphical methods. The calculus that is used in this article was introduced in \cite{ioantonia} and provides an extension of the methods in \cite{BrinkSatchler68} that is especially useful in LQG because it incorporates an easy treatment of (non-trivial) group-elements.
\subsubsection{Basic definitions}
In the following small Latin letters represent irreducible representations $j,\cdots\in\frac{1}{2}\N$ while Greek ones, $\alpha,\cdots=-j,-j+1,\cdots j$, are magnetic numbers. A state $\ket{j,\alpha}$ in the Hilbert space $\mathcal{H}_j$ is visualized by $\makeSymbol{\raisebox{.5\height}{\includegraphics{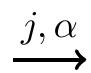}}}$, the adjoint $\bra{j,\alpha}$ is drawn as $\makeSymbol{\raisebox{.5\height}{\includegraphics{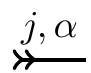}}}$ and a representation matrix $[R^{j}(g)]^{\alpha}_{\;\;\beta}$ with $g\in\SU(2)$ is pictured as
$
\makeSymbol{\raisebox{.4\height}{\includegraphics{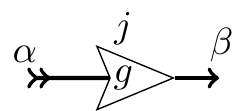}}}~.
$
Two lines carrying the same representation can be either connected by taking the scalar product, $\makeSymbol{\raisebox{.5\height}{\includegraphics{Graphics/lorentz1-figure0}}}\!\!\makeSymbol{\raisebox{.5\height}{\includegraphics{Graphics/lorentz1-figure1}}}=\sum\limits_{\alpha=-j}^j \scal{j,\alpha|j,\alpha}$, or by contraction with
\begin{align}
\label{eqn:dual}
\makeSymbol{\raisebox{.75\height}{\includegraphics{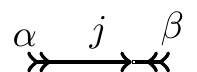}}}
			=\makeSymbol{\raisebox{.75\height}{\includegraphics{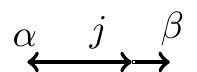}}}
			=\dtensor{j}{\beta}{\alpha}=(-1)^{j-\alpha}\delta_{\alpha,-\beta}~.
\end{align}
This two-valent intertwiner defines an isometry between the vector representation ($\tip$, tip) and the adjoint representation ($\feather$, feather) whereby the transformations
\[
\sum_{\beta} \dtensor{j}{\beta}{\alpha} \ket{j,\beta}=\bra{j,\alpha}\;\text{ and }\;
\sum_{\beta} \dtensor{j}{\alpha}{\beta}\bra{j,\beta}=\ket{j,\alpha}
\]
are graphically encrypted in $\makeSymbol{\raisebox{.7\height}{\includegraphics{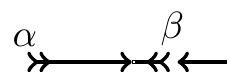}}}=\makeSymbol{\raisebox{.6\height}{\includegraphics{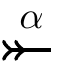}}}$ and $\makeSymbol{\raisebox{.7\height}{\includegraphics{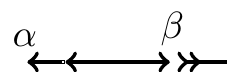}}}=\makeSymbol{\raisebox{.6\height}{\includegraphics{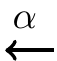}}}$. The second identity can be derived from the first one using
\beq
\label{eqn:idDual}
\makeSymbol{\raisebox{.75\height}{\includegraphics{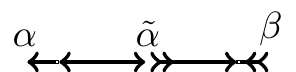}}}=\sum_{\tilde{\alpha}=-j}^j\dtensor{j}{\alpha}{\tilde{\alpha}}\dtensor{j}{\beta}{\tilde{\alpha}}
	=\delta^{\alpha}_{\beta}
	:=\makeSymbol{\raisebox{.75\height}{\includegraphics{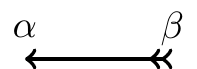}}}~.
\eeq
Note, the inner direction of \eqref{eqn:dual} is crucial here since it indicates the order of the magnetic indices that differ by a sign when it is interchanged due to $(-)^{j+\alpha}=(-)^{2j}(-)^{j-\alpha}$ and $(-)^{2(j-\alpha)}=1$. Consequently,
\begin{gather*}
\begin{gathered}
\makeSymbol{\raisebox{.35\height}{\includegraphics{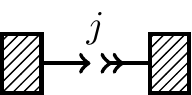}}}
\,=\,\makeSymbol{\raisebox{.35\height}{\includegraphics{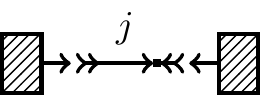}}}
=(-1)^{2j}\;\;\makeSymbol{\raisebox{.35\height}{\includegraphics{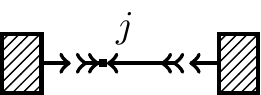}}}
= (-1)^{2j}\;\; \makeSymbol{\raisebox{.35\height}{\includegraphics{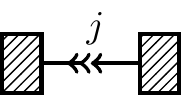}}}~.
\end{gathered}
\end{gather*}
Using \eqref{eqn:dual} and the properties of Wigner matrices it is also straightforward to prove the following identity:
\begin{gather}
\begin{gathered}
\label{eqn:invertG}
R^{\beta}\,_{\alpha}(g^{\scriptscriptstyle{-1}})=(-1)^{\alpha-\beta}R^{-\alpha}\,_{-\beta}(g)=\overline{R^\alpha\,_\beta(g)}\\[2pt]
\makeSymbol{\includegraphics{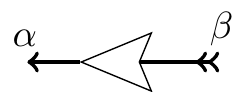}}
=\makeSymbol{\raisebox{.25\height}{\includegraphics{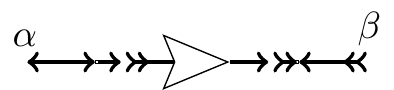}}}=
\makeSymbol{\raisebox{.25\height}{\includegraphics{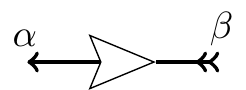}}}~.
\end{gathered}
\end{gather}
\subsubsection{Recoupling}
The basic building block of recoupling theory are the $3j$-symbols
\beq
\label{eqn:3-j}
\threej{a}{\alpha}{b}{\beta}{c}{\gamma}\equiv\makeSymbol{\includegraphics{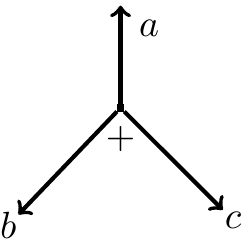}}= \makeSymbol{\includegraphics{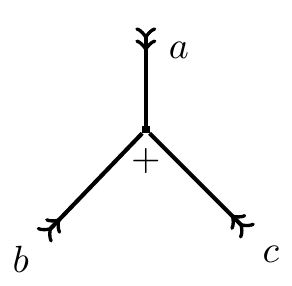}}
\eeq
that arise from coupling two irreps $a$ and $b$:
\begin{gather}
\begin{gathered}
\label{eqn:basicrule}
	(R^{a}(g))^{\alpha}\,_{\tilde{\alpha}}(R^{b}(g))^{\beta}\,_{\tilde{\beta}}
	=\sum_{c=|a-b|}^{a+b}(2c+1)\sum_{\gamma,\tilde{\gamma}=-c}^{c}
	\threej{a}{\tilde{\alpha}}{b}{\tilde{\beta}}{c}{\tilde{\gamma}}
	\threej{a}{\alpha}{b}{\beta}{c}{\gamma}\overline{(R^{c}(g))^{\gamma}\,_{\tilde{\gamma}}}\\[3pt]
\makeSymbol{\includegraphics{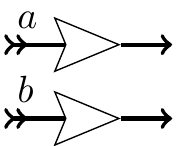}}\;\;=\sum_{c=|a-b|}^{a+b}(2c+1)
\makeSymbol{\includegraphics{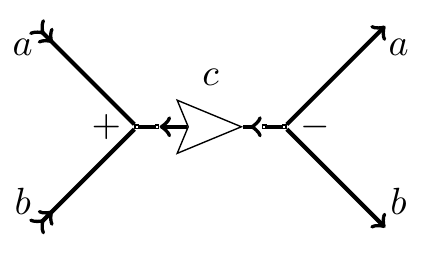}}~.
\end{gathered}
\end{gather}
As the tensor product $\mathcal{H}_a\otimes\mathcal{H}_b$ decomposes into $\bigoplus\limits_{c=|a-b|}^{a+b}\mathcal{H}_c$, a $3j$-symbol \eqref{eqn:3-j} is vanishing if $(a,b,c)$ are not compatible, i.e. if they don't obey $|a-b|\leq c\leq a+b$, and if $\alpha+\beta+\gamma\neq0$. Since $3j$'s are intertwiner (or equivariant maps) they commute with the group action. That means $[R^{a}(g)]^{\alpha}_{\alpha'}\iota^{\alpha'\beta\gamma}=\iota^{\alpha\beta'\gamma'} [R^{b}(g)]^{\beta}_{\beta'}[R^{c}(g)]^{\gamma}_{\gamma'}$ for any trivalent intertwiner $\iota^{\alpha\beta\gamma}$, which is encoded in 
\beq
\label{eqn:equivariant}
\makeSymbol{\includegraphics{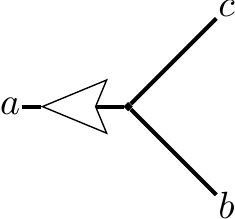}}
=\;
\makeSymbol{\includegraphics{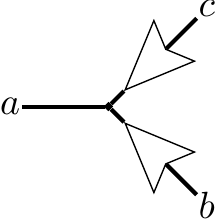}}~.
\eeq
Apart from that, these symbols are invariant under an even permutation of columns and related by $(-)^{a+b+c}$ for an odd permutation. For this reasons one has to assign an orientation to the nodes labeling anti-clockwise ($+$) and clockwise ($-$) order of the links (compare with \eqref{eqn:3-j}). Another important symmetry is 
\beq
\label{eqn:sym3j-2}
\threej{a}{\alpha}{b}{\beta}{c}{\gamma}=(-1)^{a+b+c}\threej{a}{-\alpha}{b}{-\beta}{c}{-\gamma}~.
\eeq
The trivalent intertwiners \eqref{eqn:3-j} are fundamental building blocks. All other intertwiners can be obtained from these ones. For example a node which has in- as well as outgoing links is constructed by multiplying with \eqref{eqn:dual}:
\begin{align}
\begin{split}
\label{eqn:ingoing}
\makeSymbol{\includegraphics{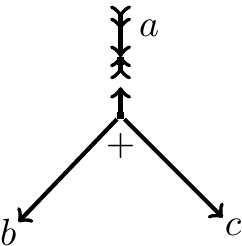}}
=(-1)^{a-\alpha}\threej{a}{-\alpha}{b}{\beta}{c}{\gamma}:=
\makeSymbol{\includegraphics{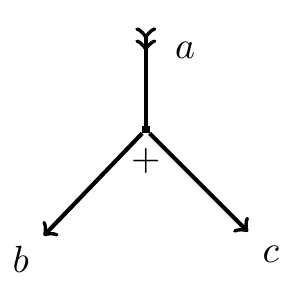}}~.
\end{split}
\end{align}
This together with \eqref{eqn:sym3j-2} proves the equivalence of a trivalent nodes whose links are all ingoing with one whose links are all outgoing, as it was claimed in \eqref{eqn:3-j}. Moreover, the intertwiner \eqref{eqn:ingoing} is of importance when coupling two holonomies with opposed orientation. In this case one finds with \eqref{eqn:invertG} and $d_c\equiv 2c+1$
\begin{align}
\label{eqn:basicOpposed}
\makeSymbol{\includegraphics{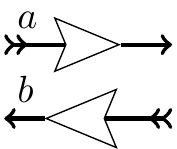}}\;\;\overset{\eqref{eqn:invertG}}{=}\;\;
\makeSymbol{\includegraphics{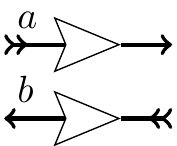}}\;
=(-1)^{2a}\sum_{c=|a-b|}^{a+b}d_{c}
\makeSymbol{\includegraphics{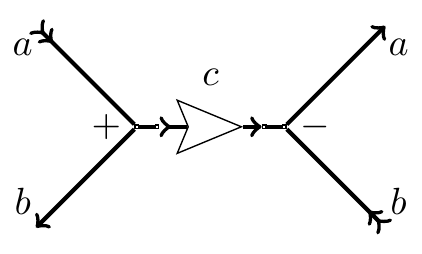}}~.
\end{align}
A four-valent node arises from the contraction of two $3j$-symbols:
\beq
\label{eqn:4-valent}
\sum_{\chi=-x}^{x} \threej{a}{\alpha}{b}{\beta}{x}{\chi}(-1)^{x-\chi}\threej{x}{-\chi}{c}{\gamma}{d}{\delta}=\makeSymbol{\includegraphics{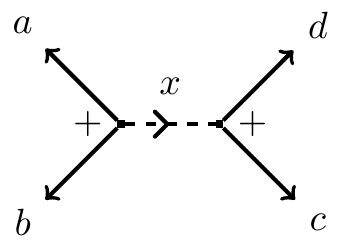}}~.
\eeq
The internal leg $x$ is drawn as a dashed line to emphasize that it is not a 'true' edge in the sense that the leg does not have a real extension in $\Sigma$ but corresponds to a point and consequently can not cary holonomies. Higher valent nodes are obtained similarly by adding more and more internal (dashed) lines.\\
The advantage of the intertwiners built above is that they provide an orthogonal basis in the space of invariants $\mathrm{Inv}[\bigotimes\limits_{\text{ingoing}} \mathcal{H}^{\ast}_j\otimes\bigotimes\limits_{\text{outgoing}} \mathcal{H}_k]$ where $^{\ast}$ denotes the adjoint. The $3j$-symbols are even normalized to one:  
\begin{gather}
\label{eqn:firstorthogonality}
\makeSymbol{\includegraphics{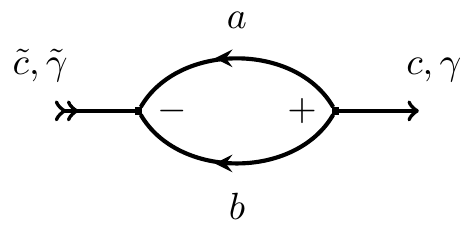}}=\sum_{\alpha,\beta}\threej{a}{\alpha}{b}{\beta}{\tilde{c}}{\tilde{\gamma}}\threej{a}{\alpha}{b}{\beta}{c}{\gamma}=\frac{1}{d_c} \delta_{c,\tilde{c}}\;\delta^{\tilde{\gamma}}\,_{\gamma}=\frac{1}{d_c}
\makeSymbol{\raisebox{.75\height}{\includegraphics{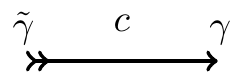}}}\\\nonumber
\implies\quad\makeSymbol{\raisebox{0\height}{\includegraphics{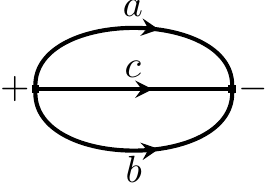}}}\; =1
~.
\end{gather}
While $\mathrm{Inv}[\mathcal{H}_a\otimes\mathcal{H}_b\otimes\mathcal{H}_c]$ is one dimensional the dimension of an $n$-valent intertwiner space equates the number of possible labelings of the internal legs. For example, $x$ in \eqref{eqn:4-valent} subscripts a basis\footnote{This basis is orthogonal but \emph{not normalized to one}. The contraction of two intertwiners $\iota_x,\iota_y$ yields $\frac{\delta_{x, y}}{d_x}$ (signs can be always absorbed by adjusting orientations) which follows directly from \eqref{eqn:firstorthogonality}.} of the four-valent space. Instead of coupling $a,b$ and $d,c$ in \eqref{eqn:4-valent} it is equally well allowed to chose a different coupling scheme, e.g $b,c--d,a$, corresponding to a different basis. These different bases are related by \emph{$6j$-symbols}:
\beq
\label{eqn:6j}
\makeSymbol{\includegraphics{Graphics/lorentz1-figure32}}=\sum_{y} d_y\sixj{a}{c}{b}{d}{x}{y} \makeSymbol{\includegraphics{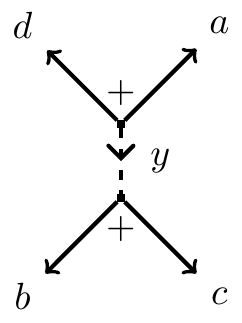}}~.
\eeq
In a similar manner coupling schemes of five spins are related by $9j$-symbols (see appendix \ref{app:recoupling}), schemes of six spins are related by $12j$'s and so forth.

\subsubsection{Simplifications}
\label{ssec:simplGraph}
To avoid unnecessary complications a slightly simplified version of the calculus introduced above is used henceforth. We forgo to display magnetic indices if not explicitly necessary as we did before, abbreviate $\makeSymbol{\includegraphics{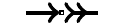}}$ by $\makeSymbol{\includegraphics{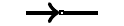}}$. In addition to that vertices are assumed to have anti-clockwise order. Only clockwise orientation is marked explicitly by a label $-$ (sometimes $+$-signs are kept for clarity). Furthermore, true edges of a spin-net are simply drawn as solid lines without explicitly showing the dependence on the group elements (triangles). Like above dashed lines are helplines to follow the coupling at a node and do not have real extension in $\Sigma$.\\
Even though this simplifies the diagrams there are certain aspects which have to be respected when evaluating the action of an operator on a spin-net. Obviously holonomies are always assigned to edges in $\Sigma$ so that they can only be coupled via \eqref{eqn:basicrule} or \eqref{eqn:basicOpposed} if they share (a part of) a solid edge. Contra wise, dashed parts can be modified arbitrarily or even removed (see below for examples) as long as they live at the same node\footnote{The only case where a solid line can be removed is when through the coupling of other holonomies the coloring of the edge is changed to the trivial representation.}.\\
\label{ex:simplGraph}
Nevertheless, solid parts can be transformed into dashed ones by the action of an operator as will be explained in the following example. Consider the action of $\Tr[\cdots h^{(m)}_{s_e}\hat{O}h^{(m)}_{s_e}\,^{-1}]$ on a node $A$ (box in the graphic below) where $s_e$ is a part of an edge $e$ emanating from the node and $\hat{O}$ is an arbitrary operator not depending on a holonomy along $s_e$. To start with, the first holonomy $h^{(m)}_{s_e}\,^{-1}$ is coupled to $s_e$ via \eqref{eqn:basicOpposed}:
\[
\makeSymbol{\raisebox{.3\height}{\includegraphics{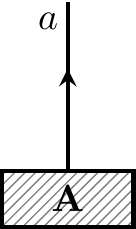}}}\quad
\xrightarrow{\;\;h^{(m)}_{s_e}\,^{-1}\;\;} \quad
\sum_{b} d_{b}\, (-)^{2 a}\;
\makeSymbol{\raisebox{.3\height}{\includegraphics{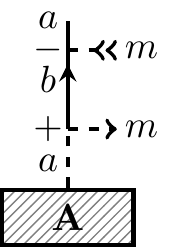}}}~.
\]
Here $\tip\, m$ and $m \feather$ represent the magnetic indices of $h^{(m)}_{s_e}\,^{-1}$ which are not contracted yet. The operator $\hat{O}$ will in general change the intertwiner, which now consist  of $A$ \emph{and} the dashed edge $a$. After evaluating $\hat{O}$ the holonomy $h^{(m)}_{s_e}$ is multiplied as the tip of $h^{(m)}_{s_e}$ is contracted with the feather of $h^{(m)}_{s_e}\,^{-1}$ so that
\[
\cdots\xrightarrow{\;h^{(m)}_{s_e}\hat{O}\;}\;
\sum_{b} d_{b}\,(-)^{2a}\sum_{c} O^a_c\;
\makeSymbol{\raisebox{.3\height}{\includegraphics{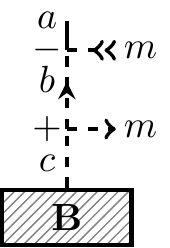}}}
\quad\xrightarrow{\;\Tr[\cdots\;}\quad
\sum_{b} d_{b}\,(-)^{2a}\sum_{c} O^a_c \;
\makeSymbol{\raisebox{.3\height}{\includegraphics{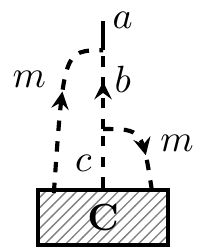}}}~.
\]
The solid part $b$ is turned into a dashed part since first of all going along $s_e^{-1}$ and then along $s_e$ pulls $m\, \feather$ back to the node and secondly the group element can be removed of $h^{(m)}_{s_e}$ and $b$ due to \eqref{eqn:equivariant}. After all other components of the operator have been applied the trace can be closed merging the remaining $\tip$ and $\feather$.  

\subsection{Important identities}
\label{ssec:Imp_Id}
It will now be demonstrated how to work with the above calculus by means of specific examples, which will be important in what follows. A useful technique to simplify complicated couplings is to insert a resolution of identity at the intertwiner space. For example contracting the vertex on the left of equation
\beq
\label{eqn:6jmove}
\makeSymbol{\includegraphics[scale=0.8]{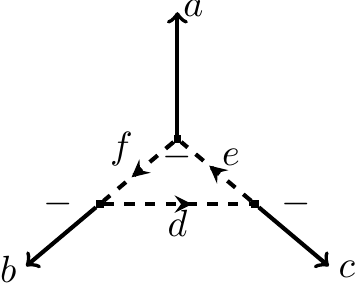}}
=\makeSymbol{\includegraphics[scale=0.8]{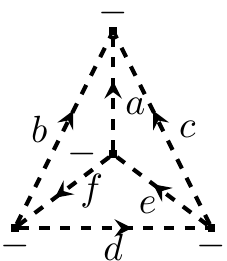}}\quad
\makeSymbol{\includegraphics[scale=0.8]{Graphics/lorentz1-figure23}}
\eeq
with a trivalent vertex yields a tetrahedron that equals $\tinysixj{a}{d}{b}{e}{c}{f}$. This relation in combination with \eqref{eqn:6j} proves to be very handy for the evaluation of the following diagram:
\begin{gather}
\nonumber
\makeSymbol{\includegraphics[scale=0.8]{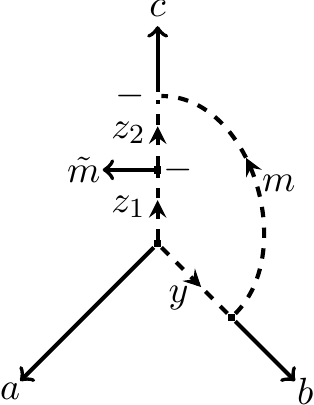}}
\overset{\eqref{eqn:6j}}{=} 
\sum_{z_3} d_{z_3}(-)^{2 z_1} (-)^{z_3+c+\tilde{m}}\sixj{z_1}{c}{z_2}{z_3}{\tilde{m}}{m} 
\makeSymbol{\includegraphics[scale=0.8]{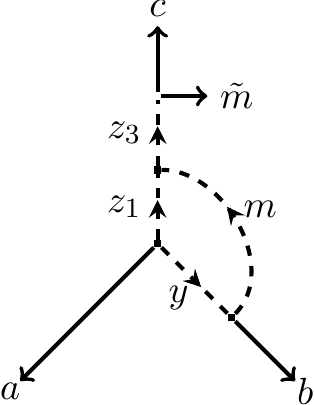}}\\
\label{eqn:move1}
\overset{\eqref{eqn:6jmove}}{=}
\sum_{z_3} d_{z_3} (-)^{2y}(-)^{a+b+c+\tilde{m}}\sixj{z_1}{c}{z_2}{z_3}{\tilde{m}}{m} \sixj{a}{m}{b}{z_1}{z_3}{y} 
\makeSymbol{\includegraphics[scale=0.8]{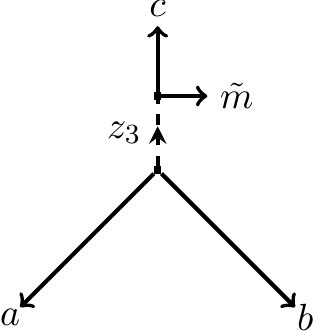}}
\end{gather}
The signs arise from adjusting the orientation of edges and nodes. E.g. to apply \eqref{eqn:6jmove} the orientation of the link $z_3$ in the second graphic has to be flipped and all nodes involved must be labeled by $-$ instead of $+$ that gives $(-)^{4 z_3+2(y+m)} (-)^{a+b+z_3}=(-)^{2(y+m)} (-)^{a+b+z_3}$. Since $4j$ is an even integer for any spin $j$ and $(-)^{2z_1}=(-)^{2(m+z_3)}$ because $(z_1,z_3,m)$ are compatible, the final sign reduces to $(-)^{2y}(-)^{a+b+c+\tilde{m}}$. \\
Instead of interchanging $m$ and $\tilde{m}$ one could have also used \eqref{eqn:6j} to move $\tip\, \tilde{m}$ to the edge $a$ before applying \eqref{eqn:6jmove}. Following this procedure one finds
\beq
\label{eqn:move2}
\makeSymbol{\includegraphics[scale=0.8]{Graphics/lorentz2-figure0}}
=\sum_{x}d_x(-)^{x+y+z_1+z_2+m+b-a} \sixj{z_1}{a}{z_2}{x}{\tilde{m}}{y}\sixj{x}{m}{b}{z_2}{c}{y}
\makeSymbol{\includegraphics[scale=0.8]{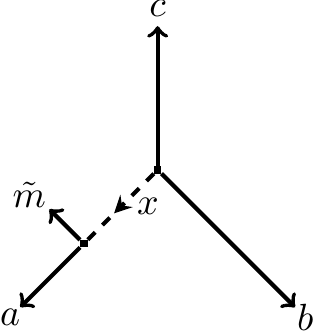}}~.
\eeq
On the other hand, the basis must be changed at least two times if $\tip\,\tilde{m}$ should be finally aligned to $b$. Thus,
\beq
\label{eqn:move3}
\makeSymbol{\includegraphics[scale=0.8]{Graphics/lorentz2-figure0}}=\sum_{\tilde{y}}d_{\tilde{y}}(-)^{2(a+y)}(-)^{b+\tilde{y}+\tilde{m}} \ninej{k}{a}{\tilde{y}}{m}{y}{b}{z_2}{z_1}{\tilde{m}}
\makeSymbol{\includegraphics[scale=0.8]{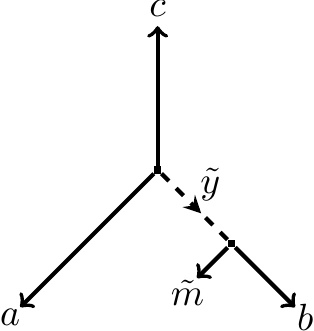}}
\eeq
where the sum over the resulting three $6j$'s (one from \eqref{eqn:6jmove}) can be summarized in a \emph{$9j$-symbol}\footnote{See appendix \ref{app:recoupling} for more details} :
\beq
\label{9jdef}
\ninej{a}{f}{r}{d}{q}{e}{p}{c} {b}:=\sum_x d_x (-1)^{2x} \sixj{a}{c}{b}{d}{x}{p} \sixj{c}{e}{d}{f}{x}{q} \sixj{e}{a}{f}{b}{x}{r}~.
\eeq
It is often easier to follow the signs when one uses \eqref{eqn:basicrule} (or \eqref{eqn:basicOpposed}) and \eqref{eqn:6jmove} instead of \eqref{eqn:6j} as it is demonstrated in the subsequent example\footnote{Since in this specific case $\tilde{m}$ is an integer, changing the orientation of the link $\tilde{m}$ does not contribute a sign which is why there is no arrow assigned to the link.}:
\[
\makeSymbol{\includegraphics[scale=.85]{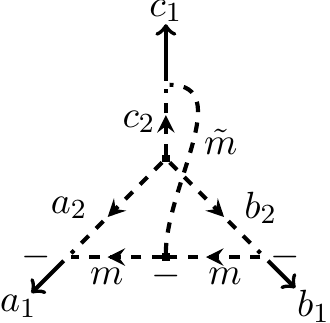}}
\overset{ \eqref{eqn:basicrule}}{=}\sum_{x} d_{x}
\makeSymbol{\includegraphics[scale=.85]{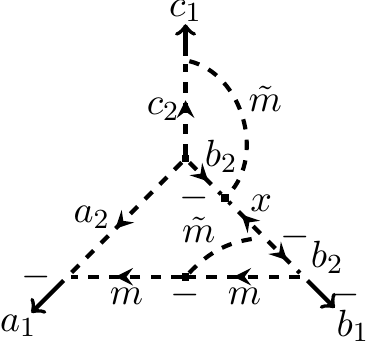}}
\]
Removing the line $\tilde{m}$ between the upper link and the one on the right leads indeed to the same $6j$ as \eqref{eqn:6j}. After eliminating the legs labeled by $\tilde{m}$ in the right diagram the link $m$ is vanishing as well due to \eqref{eqn:6j}. The resulting $6j$'s can again be summed up to a $9j$:
\begin{gather}
\label{eqn:trick1}
\makeSymbol{\includegraphics[scale=.85]{Graphics/lorentz2-figure5}}
=(-1)^{b_1+b_2+\tilde{m}+m}
\ninej{c_1}{\tilde{m}}{c_2}{b_1}{m}{b_2}{a_1}{m}{a_2}
\makeSymbol{\includegraphics[scale=.8]{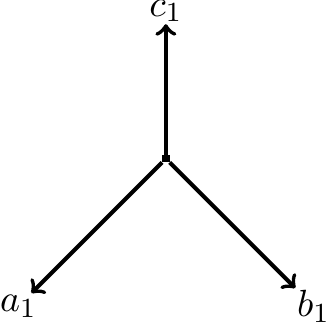}}~.
\end{gather}
Apart from the above examples couplings of the form
\[
\makeSymbol{\includegraphics[scale=0.8]{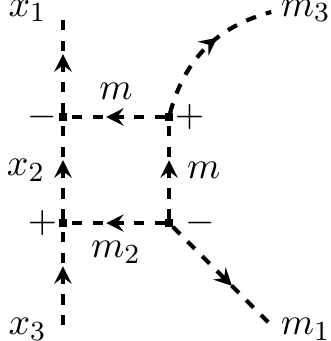}}
\]
are encountered in the subsequent analysis. Coupling the link $m_3$ in a suitable way (here with \eqref{eqn:basicOpposed}) yields 
\beq
\label{eqn:trick2}
\sum_{m_4,x_{4}} d_{m_4}d_{x_4}(-)^{2x_{3}}
\makeSymbol{\includegraphics[scale=0.8]{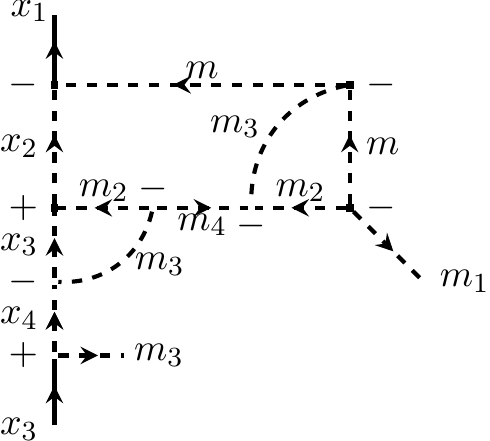}}=
\sum_{x_4}d_{x_4}(-)^{x_3+x_4+m_2+m_1+m}
\ninej{m_2}{x_2}{x_3}{m_1}{x_1}{x_4}{m}{m}{m_3}
\makeSymbol{\includegraphics[scale=0.8]{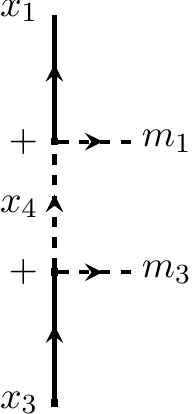}}
\eeq
where $\eqref{eqn:6jmove}$ was utilized first to take away the two legs $m_3$ and afterwards $x_2$.

\subsection{Action of the volume}
\label{ssec:volume}
Even though the volume can not be computed analytically for generic configurations it is comparatively easy to calculate it for gauge variant trivalent vertices transforming in a low spin. This is exactly the type of nodes that are of interested here. Nevertheless, we will not explicitly calculate this matrix elements in this section but only summarize some generic properties. More details can be found in appendix \ref{app:Volume}. \\ 
Non-invariant trivalent intertwiners can be treated as four-valent (invariant) ones, e.g.
\[
\makeSymbol{\includegraphics[scale=0.7]{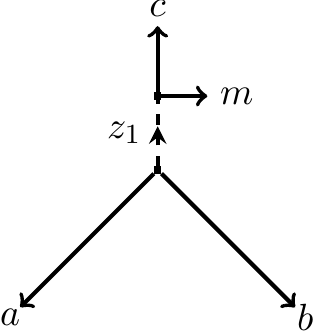}}~,
\]
where one leg, here $m$, does not correspond to a true edge but indicates in which representation the node is transforming. In a slight abuse of the conventions these open ends $m$ are drawn as solid rather than dashed lines for better visibility. Whether a leg at the intertwiner corresponds to a part of a true edge or is indicating gauge variance follows from the context.\\
In general $\hat{V}$ is altering the intertwiner structure, so that 
\begin{equation}
\label{eqn:voume1}
\hat{V}
\makeSymbol{\includegraphics[scale=0.7]{Graphics/lorentz2-figure11}}
:=\sum_{z_2}V_{z_2}^{z_1}(a,b|c;m)
\makeSymbol{\includegraphics[scale=0.7]{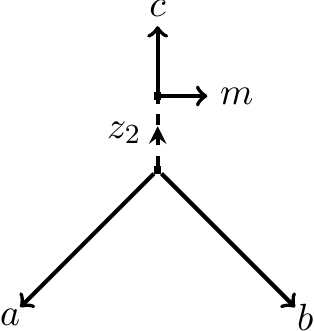}}~.
\end{equation}
Note, the above definition of the matrix elements $V_{z_2}^{z_1}$ \emph{depends crucially on the order of the node as well as the orientation of the edges}. Consider for example a node whose legs $a$ and $b$ are interchanged. Switching the legs back into the original position before acting with $\hat{V}$ gives a sign $(-)^{z_1+a+b}$ while exchanging them after the volume has been applied yields $(-)^{z_2+a+b}$ so that $V_{z_2}^{z_1}(a,b|c;m)= (-)^{z_1-z_2}V_{z_2}^{z_1}(b,a|c;m)$. 
\section{Action of the Euclidean constraint}
\label{sec:main1}
The action of the Euclidean constraint on trivalent nodes was determined for the first time in \cite{Gaul} and \cite{Borissovetal97} using Temperely-Lieb algebras and recalculated in \cite{ioantonia} by graphical methods similar to the one introduced above. In this section a powerful trick is presented that hugely simplifies this calculation.

\subsection{Action on gauge-invariant trivalent vertices}

For a single trivalent vertex with adjacent edges $e_i,e_j,e_k$ the Euclidean constraint \eqref{eqn:quantumH} reduces to  
\beq
\label{eqn:eucliTr}
\frac{4i}{N^2_m\,\kappa\, l^2_p}\epsilon^{ijk}\;\Tr_m\left[(h_{\alpha_{ij}}-h_{\alpha_{ji}})h_{s_k}[h_{s_k^{-1}},V]\right]
\eeq
where the Lapse $N(\vgraph)$ was set to one and the subscript $m$ in $\Tr_m$ indicates that the holonomies have spin $m$. Moreover, since $\Tr_m[h_{\alpha_{ij}}]=\Tr_m[(h_{\alpha_{ij}})^{-1}]=\Tr_m[h_{\alpha_{ji}}]$ expression \eqref{eqn:eucliTr} reduces further to\footnote{This argument applies to \emph{all} intertwiners so that the commutator can be always replaced by $h_{s_k}V h_{s_k^{-1}}$.}
\beq
\label{eqn:E-trace}
\frac{4i}{N^2_m\,\kappa\, l^2_p}\epsilon^{ijk}\;\Tr_m\left[(h_{\alpha_{ij}}-h_{\alpha_{ji}})h_{s_k}Vh_{s_k^{-1}}\right]~.
\eeq
To begin with, the holonomy $h_{s_k}^{-1}$ has to be coupled to the edge $e_k$ what leads to a gauge variant vertex because the magnetic indices are not contracted yet. Consequently,
\beq
\label{eqn:eh1}
V h_{s_k}^{-1}\!\!\!\makeSymbol{\includegraphics[scale=0.7]{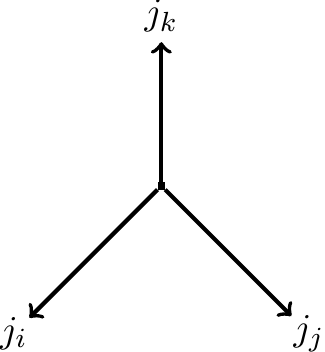}}
= V \sum_{c_1} d_{c_1} (-)^{2j_k} \!\!\!
\makeSymbol{\includegraphics[scale=0.7]{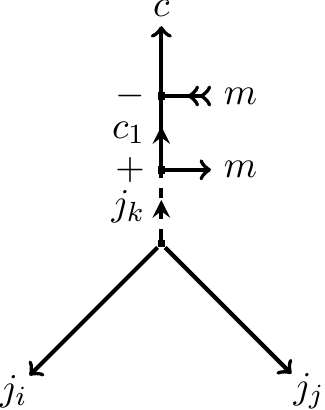}}
= \sum_{c_1,c_2} d_{c_1} (-)^{2j_k}\, V_{c_{2}}^{j_{k}}(j_i,j_j|c_{1};m)\!
\!\!
\makeSymbol{\includegraphics[scale=0.7]{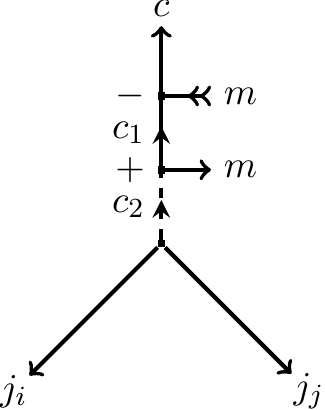}}~.
\eeq
The action of the remaining holonomies in \eqref{eqn:E-trace} can be simplified by
\begin{theorem*}[{\bf Loop trick}]
\vspace*{-5pt}\mbox{}
\beq
\vspace*{-5pt}
\label{eqn:looptrick}
\frac{1}{2}(h_{\alpha_{ij}}-h_{\alpha_{ji}})h_{s_k}= \epsilon(i,j,k) \sum_{\tilde{m}\in\left(2\mathbb{N}+1\right)}d_{\tilde{m}}\,(-)^{2m}
\makeSymbol{\includegraphics[scale=0.8]{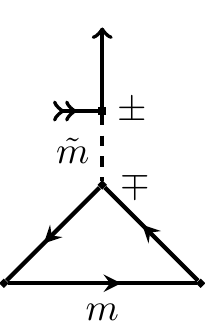}}
\eeq	
where $\epsilon(i,j,k)=1$ if the edges $e_i,e_j,e_k$ are ordered anti-clockwise, otherwise it equals  $-1$. 
\end{theorem*}
\begin{proof} 
The 'free ends' of $(h_{\alpha_{ij}}-h_{\alpha_{ji}})h_{s_k}$ can be joined artificially by \eqref{eqn:basicOpposed} what is graphically encoded in 
\[
(h_{\alpha_{ij}}-h_{\alpha_{ji}})h_{s_k}=
\epsilon(i,j,k) \left(
\makeSymbol{\includegraphics[scale=0.8]{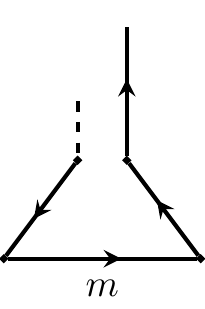}}
-
\makeSymbol{\includegraphics[scale=0.8]{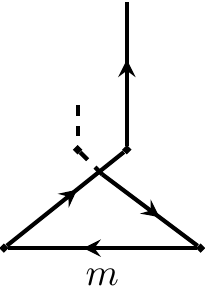}}\right)\\
=
\epsilon(i,j,k)
\sum_{\tilde{m}}d_{\tilde{m}} (-)^{2m}\left(
\makeSymbol{\includegraphics[scale=0.8]{Graphics/lorentz1-figure50}}
-
\makeSymbol{\includegraphics[scale=0.8]{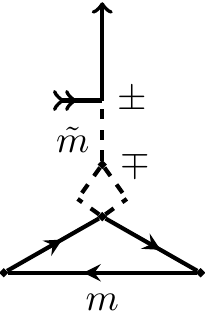}}
\right)~.
\]
Suppose $h_{ij}=[R^m(g)]^{\mu}\,_{\nu}$ then
\begin{gather*}
\makeSymbol{\includegraphics{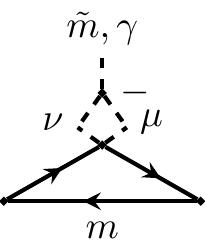}}
\equiv\sum_{\mu,\nu}\threej{\tilde{m}}{\gamma}{m}{\nu}{m}{-\mu} (-1)^{m-\mu}[R^m(g^{-1})]^{\nu}\,_{\mu}\\
=\sum_{\mu,\nu}(-1)^{2m+\tilde{m}}\threej{\tilde{m}}{\gamma}{m}{-\mu}{m}{\nu} (-1)^{m-\nu}[R^m(g)]^{-\mu}\,_{-\nu}\\
=(-1)^{\tilde{m}}\sum_{\mu,\nu}\threej{\tilde{m}}{\gamma}{m}{\mu}{m}{-\nu} (-1)^{m-\nu}[R^m(g)]^{\mu}\,_{\nu}
\equiv
(-)^{\tilde{m}}\makeSymbol{\includegraphics[scale=0.8]{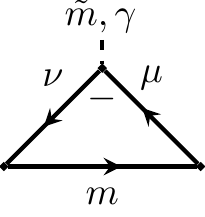}}~.
\end{gather*}
To pass from the first line to the second, relation \eqref{eqn:invertG} was invoked. The sign $(-)^{2m+\tilde{m}}$ originates from the permutation of the columns of the $3j$-symbol and is annihilated due to \eqref{eqn:sym3j-2} in the last line.\\
Inserting this in the first equality proves the theorem.
\end{proof}
As in the example on page \pageref{ex:simplGraph}, adding $h_{s_k}$ to \eqref{eqn:eh1} transforms the solid line $c_1$ into a dashed line so that with \eqref{eqn:looptrick} and $\tsum\limits_x:=\sum\limits_x d_x$ one obtains
\begin{gather}
\label{eqn:a}
\Tr_m[(h_{\alpha_{ij}}-h_{\alpha_{ji}})\cdots]
\!\!\!\makeSymbol{\includegraphics[scale=0.8]{Graphics/lorentz2-figure13}}\!\!\!\!
=2\,\epsilon(i,j,k)\!\!\!\!\!\tsum_{\stackrel{a,b,c_1;}{m_1\in\left(2\mathbb{N}+1\right)}} \!\!\sum_{c_2} 
(-)^{2b}\,V_{c_{2}}^{j_{k}}(j_i,j_j|c_1;m) \!\!\!\!\!\!\!\!\!
\makeSymbol{\includegraphics[scale=0.75]{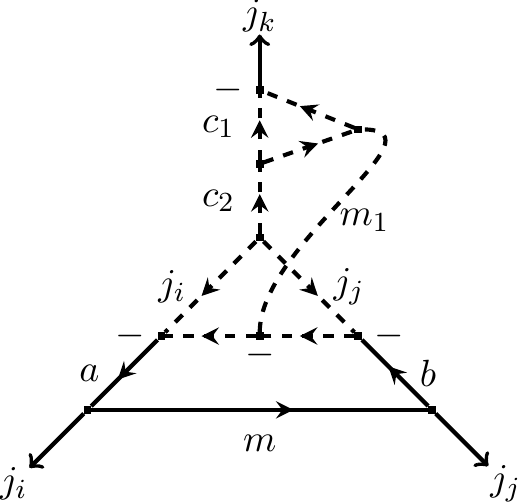}}\!\!\!\!\!\!\\
\label{eqn:b}
=2\,\epsilon(i,j,k)\!\!\!\!\!\tsum_{\stackrel{a,b,c_1;}{m_1\in\left(2\mathbb{N}+1\right)}}\!\! \sum_{c_2} 
V_{c_{2}}^{j_{k}}(j_i,j_j|c_1;m)  (-)^{j_j-b+c_2-c_1}
\sixj{m}{j_k}{m}{c_2}{m_1}{c_1}
\ninej{a}{j_i}{m}{b}{j_j}{m}{j_k}{c_2}{m_1}
\makeSymbol{\includegraphics[scale=0.75]{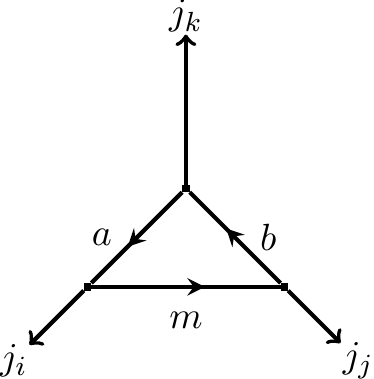}}\!\!\!\!\!\!
\end{gather}
In \eqref{eqn:a} the sign $(-)^{2(j_k+m)}$ stemming from coupling $h_{s_k^{-1}}$ and from \eqref{eqn:looptrick} is canceled by $(-)^{2m}$, which arises from coupling the part $h_{s_j^{-1}}$ of $h_{\alpha_{ij}}$, and by $(-)^{2a}=(-)^{2(b+j_k)}$ that is produced by the switch of the direction of $a$ after coupling $h_{s_i}$. The dashed parts of the diagram are removed by first making use of \eqref{eqn:6jmove} to get rid of the triangle $(m,m,c_1)$  and then of \eqref{eqn:trick1}. Parts of the signs resulting from this two moves were afterwards absorbed by an odd permutation of rows and columns of the $9j$-symbol\footnote{See appendix \ref{app:recoupling} for details}.\\
To obtain the full action of $\hat{H}$ the above expression has to be contracted with $\epsilon^{ijk}$ so that 
\begin{gather*}
\hat{H}\!\!\!\makeSymbol{\includegraphics[scale=0.7]{Graphics/lorentz2-figure13}}=
\tsum_{a,b} H^{j_i j_j}_{\,a\; b}(j_k) \!\!\!\makeSymbol{\includegraphics[scale=0.7]{Graphics/lorentz2-figure17}}\!\!\!
+\,\tsum_{b,c}H^{j_j j_k}_{\,b\; c}(j_i)\!\!\!\makeSymbol{\includegraphics[scale=0.7]{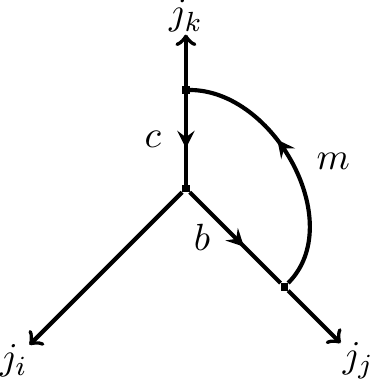}}\!\!\!
+\,\tsum_{c,a}H^{j_k j_i}_{\,c\; a}(j_j)\!\!\!\makeSymbol{\includegraphics[scale=0.7]{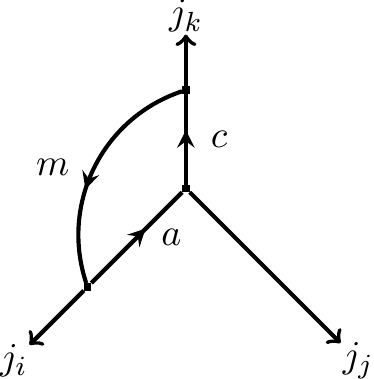}}\!\!\!
\end{gather*}
with 
\beq
\label{eqn:matH}
H^{j_i j_j}_{\,a\; b}(j_k)=\frac{8i}{N^2_m\,\kappa\, l^2_p}\;\epsilon(i,j,k)\!\!\!\!\tsum\limits_{\stackrel{c_1;}{m_1\in\left(2\mathbb{N}+1\right)}} \!\sum\limits_{c_2} 
V_{c_{2}}^{j_{k}}(j_i,j_j|c_1;m)  (-)^{j_j-b+c_2-c_1}
\tinysixj{m}{j_k}{m}{c_2}{m_1}{c_1}
\tinyninej{a}{j_i}{m}{b}{j_j}{m}{j_k}{c_2}{m_1}~.
\eeq
As for the volume it is crucial to respect the orientation of the nodes and edges in the above formula. For example:
\[
\hat{H} \makeSymbol{\includegraphics[scale=0.7]{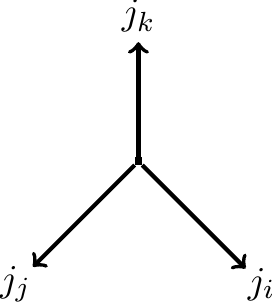}}= 
\tsum_{a,b} H^{j_j j_i}_{\,b\; a}(j_k)\makeSymbol{\includegraphics[scale=0.7]{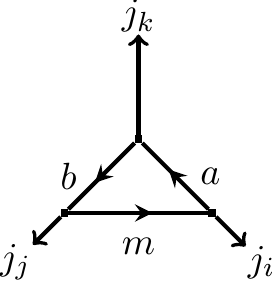}}
+\cdots
\]
Swapping the legs $j_i$ and $j_j$ back into the original position before applying $\hat{H}$ generates a sign $(-)^{j_i+j_j+j_k}$ and switching them afterwards again contributes $(-)^{a+b+j_k} (-)^{a+j_i+m}  (-)^{b+j_j+m}$ for the nodes and $(-)^{2(a+b+m)}$ for reorienting the loop so that in total\footnote{The same conclusion can be reached by directly analyzing \eqref{eqn:matH}.} $H^{j_i j_j}_{\,a\; b}(j_k)=H^{j_j j_i}_{\,b\; a}(j_k)$. This results seems to be astonishing at the first sight since $\hat{H}$ is antisymmetric. However, if $(j_i,a)$ and $(j_j,b)$ are exchanged in the matrix element then $\alpha_{ij}$ changes its direction but also the legs of the nodes are twisted so that in total this is nothing else than a relabeling and therefore should not change the value. 
\subsection{Action on gauge-variant nodes}
\label{ssec:Hvarinat}
Below we will have to calculate the extrinsic curvature of gauge-variant nodes which is why we also have to evaluate the action of $\hat{H}$ on nodes of this kind (here: a trivalent one transforming in spin $m_1$). To start with we apply the following trick:
\begin{gather}
\begin{gathered}
\label{eqn:vh}
\makeSymbol{\includegraphics[scale=0.7]{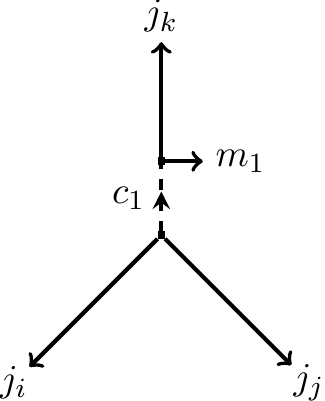}}
\xrightarrow{\;h_{s_k^{-1}}^{(m)}\;} \tsum_{z_2} (-)^{2j_k}
\makeSymbol{\includegraphics[scale=0.7]{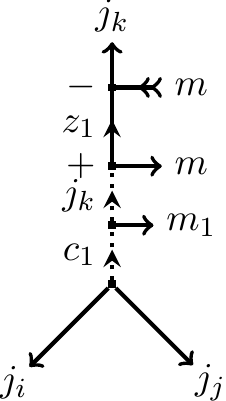}}\\
=\tsum_{z_1,m_2} (-)^{2j_k}
\makeSymbol{\includegraphics[scale=0.7]{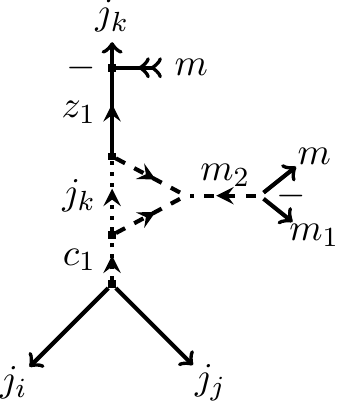}}
=\tsum_{z_1,m_2}(-)^{2 (j_k+m_1)}(-)^{c_1+z_1+m_2} \sixj{m}{c_1}{m_1}{z_1}{m_2}{j_k} 
\makeSymbol{\includegraphics[scale=0.7]{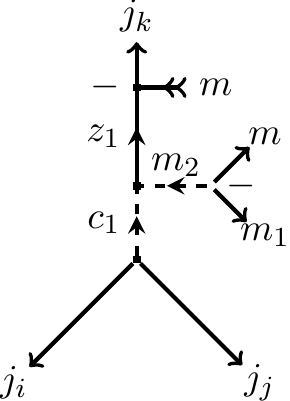}} ~.
\end{gathered} 
\end{gather}
Here, neither $\tip\,m$ nor $\tip\,m_1$ corresponds to true edges. Rather they indicate that the node is transforming in $\mathcal{H}_m\otimes\mathcal{H}_{m_1}$. Since the volume acts as a derivative operator and therefore only 'grasps'  true edges (see appendix \ref{app:Volume}) it is advisable to glue $\tip\,m$ and $\tip\,m_1$ together, taking advantage of \eqref{eqn:basicrule} and \eqref{eqn:6jmove}. Instead of having to evaluate $\hat{V}$ on the node with two open links it is now possible to simply treat it as before acting on a node with just one open leg $m_2\in \{|m-m_1|,\cdots, m+m_1\}$. Together with\eqref{eqn:looptrick} this yields
\begin{gather*}
\frac{1}{2}\Tr_m[(\alpha_{ij}-\alpha_{ji})h_{s_k}\hat{V}h_{s_k^{-1}}]\makeSymbol{\includegraphics[scale=0.7]{Graphics/lorentz2-figure20}} =\\
\epsilon(i,j,k)
\tsum_{\stackrel{z_1,m_2}{a,b}}\tsum_{m_3\in 2\N+1}\sum_{z_2}(-)^{c_1+z_1+m_2+2b} \sixj{m}{c_1}{m_1}{z_1}{m_2}{j_k}
 V_{z_2}^{c_1}\left(j_i,j_j|z_1; m_2\right)
\makeSymbol{\includegraphics[scale=0.7]{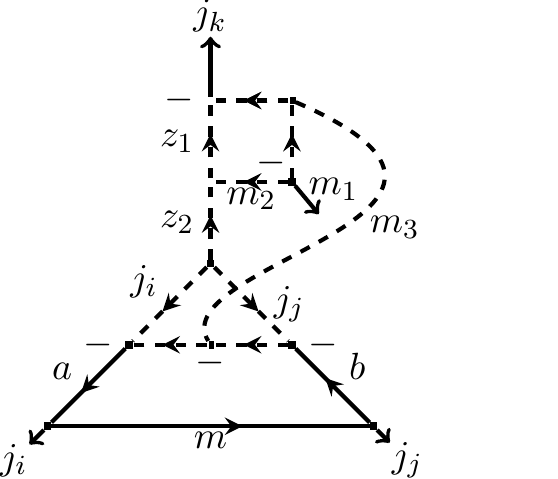}}
\end{gather*}
The dashed parts can be simplified exerting \eqref{eqn:trick1}, \eqref{eqn:trick2} and the symmetries of $9j$'s:
\begin{gather*}
\makeSymbol{\includegraphics[scale=0.7]{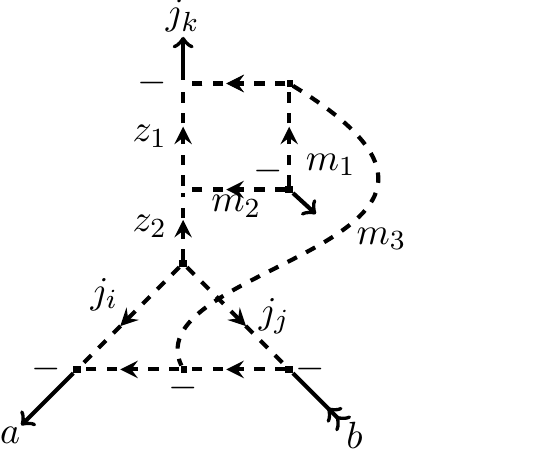}}\!\!\!\overset{\eqref{eqn:trick2}}{=}\,
\tsum_{c_2} (-)^{2z_1} (-)^{z_2+m+m_1+m_2-c_2}
\ninej{m_2}{z_1}{z_2}{m_1}{j_k}{c_2}{m}{m}{m_3}
\makeSymbol{\includegraphics[scale=0.7]{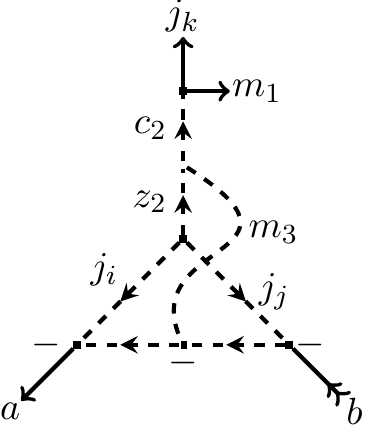}}
\\
\overset{\eqref{eqn:trick1}}{=}\tsum_{c_2} (-)^{2m} (-)^{j_i+a+m_1-m_2}
\ninej{m_2}{z_1}{z_2}{m_1}{j_k}{c_2}{m}{m}{m_3}
\ninej{j_i}{j_j}{z_2}{a}{b}{c_2}{m}{m}{m_3}
\makeSymbol{\includegraphics[scale=0.7]{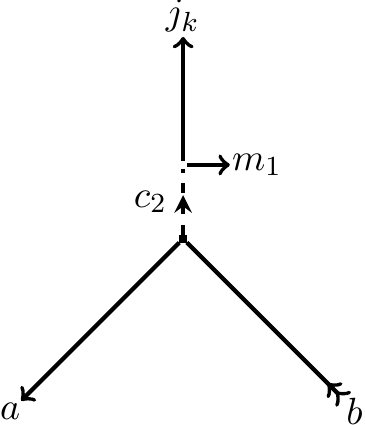}}
\end{gather*}
With the remaining summands of $H$ one can proceed similarly if the basis of the intertwiner is changed before the trace is evaluated such that $\tip \,m_1$ is always assigned to the edge facing the loop. The full amplitude is then given by 
\begin{align}
\begin{split}
&H\makeSymbol{\includegraphics[scale=0.7]{Graphics/lorentz2-figure20}}= 
\tsum_{a,b,c_2} H^{j_i j_j; c_1}_{\,a\;b;\;c_2}(j_k;m_1)
\makeSymbol{\includegraphics[scale=0.7]{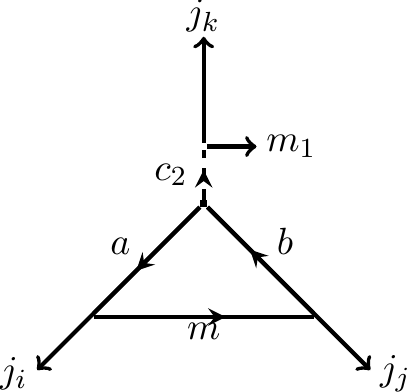}}\\
&+\tsum_{\stackrel{b,c_2}{a_1,a_2}} (-)^{2m_1}(-)^{j_i+j_j+c_1} \sixj{a_1}{c_1}{j_j}{m_1}{j_k}{j_i}H^{j_j j_k; a_1}_{b\,c_2;\;a_2}(j_i;m_1)
\makeSymbol{\includegraphics[scale=0.7]{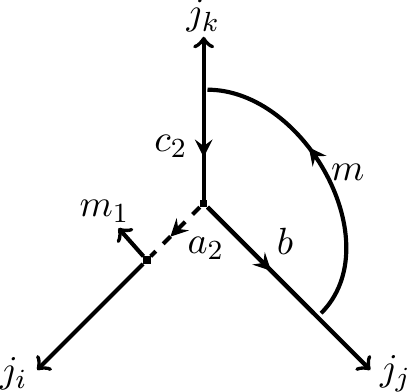}}\\
&+\tsum_{\stackrel{c_2,a}{b_1,b_2}} (-)^{2m_1}(-)^{j_i+b_1+j_k} \sixj{j_i}{m_1}{b_1}{c_1}{j_k}{j_j}H^{j_k j_i; b_1}_{c_2\,a;\;b_2}(j_i;m_1)
\makeSymbol{\includegraphics[scale=0.7]{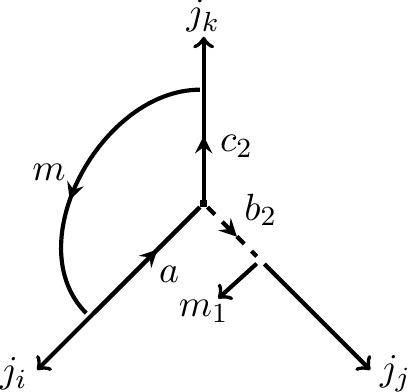}}
\end{split}
\end{align}
where 
\begin{gather*}
H^{i \,j; x}_{a\, b; y}(k;m_1)=\\
\epsilon(i,j,k)\frac{8i}{N^2_m\,\kappa\, l^2_p}\!\!\tsum_{\stackrel{z_1,m_2}{m_3\in 2\N+1}}\sum_{z_2} 
(-)^{m_1+z_1+x+i+a+2j}
\tinysixj{m}{x}{m_1}{z_1}{m_2}{k}
 V_{z_2}^{x}\left(i,j|z_1; m_2\right)
\tinyninej{m_2}{z_1}{z_2}{m_1}{k}{y}{m}{m}{m_3}
\tinyninej{i}{j}{z_2}{a}{b}{y}{m}{m}{m_3}~.
\end{gather*}
and $H^{i \,j; x}_{a\, b; y}=H^{j \,i; x}_{b\, a; y}$.\\
This concludes the discussion of the action of the Euclidean constraint on trivalent vertices, invariant as well as variant ones. The matrix elements for higher valent nodes can be in principle obtained by analogous methods (see e.g. \cite{ioantonia}  for four-valent vertices). 
\section{Matrix elements of $\hat{T}$}
\label{sec:main2}
Having all tools at our hand we can now proceed to calculate the matrix elements of
\[
\hat{T}\propto\epsilon_{ijk}\Tr_m[h_{s_i}[h_{s_i^{-1}},\hat{K}] h_{s_j} [h_{s_j^{-1}},\hat{K}] h_{s_k}[h_{s_k^{-1}},\hat{V}]]~.
\]
Because the volume and therefore also the extrinsic curvature $\hat{K}=\frac{i}{l^2_p\gamma} [\hat{V},\hat{H}]$ of a gauge-invariant trivalent node is zero, the only non-vanishing contribution of $\hat{T}$ on such nodes is proportional to\footnote{This is of course not true in general.} 
\[
\epsilon_{ijk}\Tr_m[h_{s_i}\hat{K}h_{s_i^{-1}} h_{s_j} [h_{s_j^{-1}},\hat{K}] h_{s_k}\hat{V}h_{s_k^{-1}}]~.
\]
The only unknown part in this expression is the extrinsic curvature that will be discussed next before evaluating the full trace in \ref{ssec:trace1} and \ref{ssec:trace2}. 

\subsection{Extrinsic curvature of gauge-variant trivalent nodes}

As the extrinsic curvature is linear in $\hat{H}$ its action on a trivalent (variant) vertex decomposes into a sum,
\beq
\hat{K}=\frac{i}{l^2_p\gamma}\left( \hat{K}(\alpha_{ij})+ \hat{K}(\alpha_{jk})+ \hat{K}(\alpha_{ki})\right)~, 
\eeq
of the three contributions, $\hat{K}(\alpha_{ij}):=\frac{4 i}{N_m^2\kappa l_p^2}[\hat{V},\Tr_m[(h_{\alpha_{ij}}-h_{\alpha_{ji}})h_{s_k}Vh_{s_k^{-1}}]$, associated to the loops $\alpha_{ij}$. By combining the results of section \ref{ssec:volume} and \ref{ssec:Hvarinat} one finds immediately
 \[
 K^{j_i j_j;c_1}_{\,a\;b;c_2}(j_k;m_1)=\sum_{x,y} d_y \,H^{j_i j_j;x}_{\,a\;b;\,y}(j_k; m_1)\,
 \left[\frac{1}{d_{c_2}}V^{y}_{c_2}(a,b,j_k; m_1)\, \delta_{x,c_1}
 -V^{c_1}_x(j_i,j_j|j_k;m_1) \,\delta_{y,c_2}\right]~.
 \]
for
\beq
\label{eqn:K_alpha}
 \hat{K}(\alpha_{ij})\makeSymbol{\includegraphics[scale=0.8]{Graphics/lorentz2-figure20}}
: =\tsum_{a,b,c_2} K^{j_i j_j;c_1}_{\,a\;b;\,c_2}(j_k;m_1)
\makeSymbol{\includegraphics[scale=0.8]{Graphics/lorentz2-figure28}}~.
 \eeq
and the full extrinsic curvature is given by
 \begin{gather}
 \label{eqn:Kmatrix}
\hat{K}\!\!\!
\makeSymbol{\includegraphics[scale=0.7]{Graphics/lorentz2-figure20}}=
\tsum_{a_1}(-)^{2m_1}(-)^{j_i+j_j+c_1} \tinysixj{j_i}{j_k}{j_j}{m_1}{c_1}{a_1}
\tsum_{a_2,b,c_2}  K^{j_j j_k;a_1}_{\,b c_2;\;a_2}(j_i;m_1)\!\!
\makeSymbol{\includegraphics[scale=0.7]{Graphics/lorentz2-figure29}}+\\\nonumber
\tsum_{b_1}(-)^{j_i+b_1+j_k+2m_1} \tinysixj{j_i}{m_1}{j_j}{j_k}{c_1}{b_1}
\tsum_{a,b_2,c_2}  K^{j_k j_i;b_1}_{c_2\,a;b_2}(j_j;m_1)
\makeSymbol{\includegraphics[scale=0.7]{Graphics/lorentz2-figure30}}\!\!\!\!\!\!
+\tsum_{a,b,c_2}  K^{j_i j_j;c_1}_{\,a\;b;\;c_2}(j_k;m_1) \!\!\!\!\!\!\!
\makeSymbol{\includegraphics[scale=0.7]{Graphics/lorentz2-figure28}}\!\!\!\!\!.
\end{gather}
Note, due to the symmetries of volume and Euclidean constraint $K^{j_i j_j;c_1}_{\,a\;b;\,c_2}=(-)^{c_1-c_2}K^{j_j j_i;c_1}_{\,b\;a;\,c_2}$.

\subsection{Matrix elements of $\Tr_m[h_{s_i} \hat{K} h_{s_i^{-1}} h_{s_j} \hat{K} h_{s_j^{-1}} h_{s_k}V h_{s_k^{-1}}]$}
\label{ssec:trace1}

The missing link to write down the complete action of $\hat{T}$ on trivalent invariant nodes are the contributions $\Tr_m[h_{s_i} \hat{K} h_{s_i^{-1}} h_{s_j} \hat{K} h_{s_j^{-1}} h_{s_k}V h_{s_k^{-1}}]$ and $\Tr_m[h_{s_i} \hat{K} h_{s_i^{-1}} \hat{K} h_{s_k}V h_{s_k^{-1}}]$. The latter contribution is analyzed in the succeeding section while here the action of the first trace is evaluated. \\
The term $h_{s_k}\hat{V}h_{s_k^{-1}}$ is just the same as for the Euclidean Hamiltonian. Therefore,
\begin{gather}
\label{eqn:A1}
\makeSymbol{\includegraphics[scale=0.7]{Graphics/lorentz2-figure13}} \xrightarrow{h_{s_j^{-1}}h_{s_k} \hat{V}h_{s_k^{-1}}}
\tsum_{c_1,b_1} \sum_{c_2} (-)^{2(j_k+j_j)}\, V^{j_k}_{c_2}(j_i,j_j|c_1;m)
\makeSymbol{\includegraphics[scale=0.7]{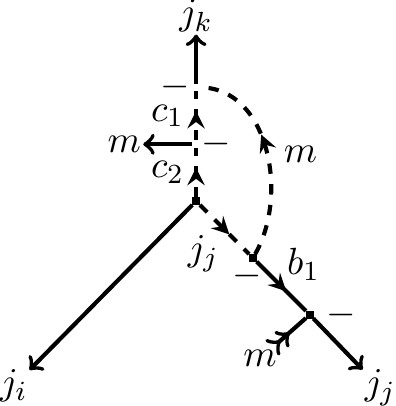}}~.
\end{gather}
Before acting with $\hat{K}$ the dashed leg $m$ must be erased from the diagram and $\tip\,m$ moved to the appropriate place. This can be done simultaneously via  \eqref{eqn:move1}, \eqref{eqn:move2} and \eqref{eqn:move3} so that $\hat{K}$ acts on the above expression by
\begin{gather*}
 \tsum_{c_1,b_1} \sum_{c_2} V^{j_k}_{c_2}(j_i,j_j|c_1; m)
 \Big[\tsum_{\stackrel{a_1,b_2}{c_3,c_4}} (-)^{j_i-j_j-j_k}
 \sixj{c_1}{c_3}{c_2}{j_k}{m}{m} \sixj{j_i}{m}{j_j}{c_3}{c_2}{b_1}
K^{j_i b_1;c_3}_{a_1b_2;c_4}(j_k;m) 
\makeSymbol{\includegraphics[scale=0.7]{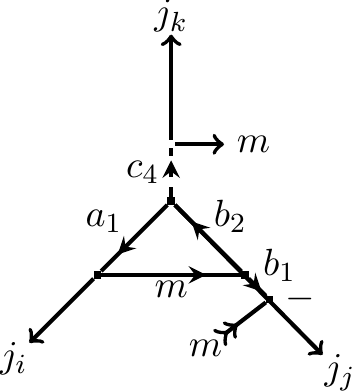}}\\[3pt]
+\tsum_{\stackrel{a_1,a_2}{b_2,c_3}} (-)^{j_i+a_1+c_1+c_2} 
\sixj{c_1}{j_i}{c_2}{a_1}{m}{j_j} \sixj{a_1}{m}{j_j}{j_k}{c_1}{b_1}
K^{b_1 j_k;a_1}_{b_2 c_3;a_2}(j_i;m) 
\makeSymbol{\includegraphics[scale=0.7]{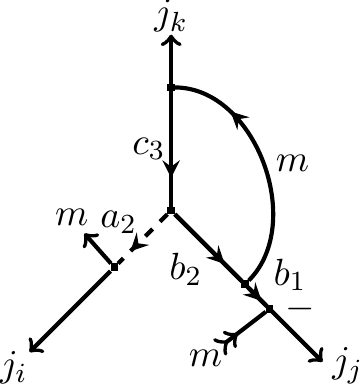}}\\[3pt]
+\tsum_{\stackrel{a_1,c_3}{b_2,b_3}} (-)^{j_j+b_2} 
\ninej{j_i}{b_2}{j_k}{j_j}{b_1}{m}{c_2}{m}{c_1}
K^{j_k j_i;b_2}_{c_3 a_1;b_3}(j_j;m)
\makeSymbol{\includegraphics[scale=0.7]{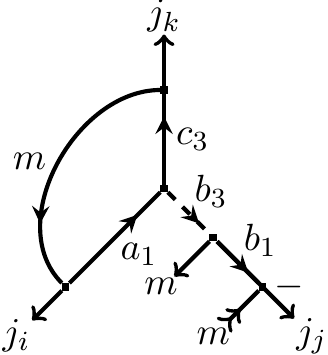}}
\Big]
\end{gather*}
In the next step $h_{s_i^{-1}} h_{s_j}$ is added. For the third term this is straightforward: $h_{s_j}$ transforms $b_1$ into a dashed line and $h_{s_i^{-1}}$ can be coupled as usual via \eqref{eqn:basicOpposed} resulting in
\beq
\label{eqn:third}
\tsum\cdots \tsum_{a_2}
\makeSymbol{\includegraphics[scale=0.7]{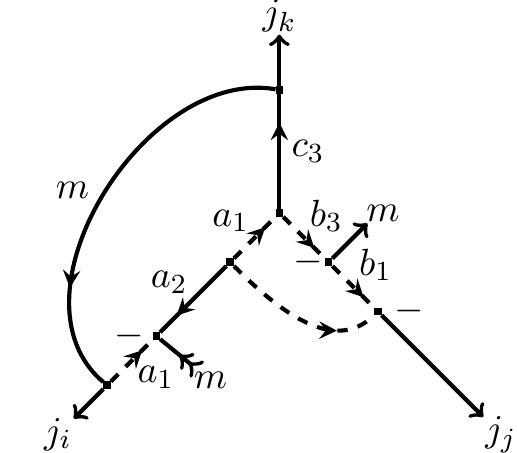}}~.
\eeq
The most efficient way to proceed with the other two terms is to couple $h_{s_j}$ via \eqref{eqn:basicrule} and then use \eqref{eqn:6jmove}. This yields
\beq
\label{eqn:first}
\tsum\cdots \tsum_{a_2,b_3}(-)^{2(a_1+b_2)}\!\!\!\!\!\!\!
\makeSymbol{\includegraphics[scale=0.7]{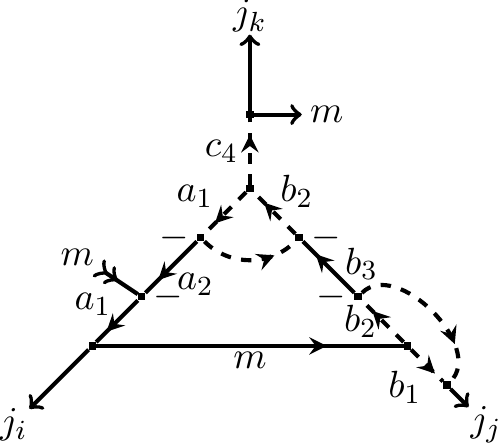}}\!\!\!\!\!\!\!
=
\tsum\cdots \tsum_{a_2,b_3}(-)^{j_j+c_4-a_1}
\tinysixj{b_1}{b_3}{b_2}{j_j}{m}{m}  \tinysixj{a_1}{b_3}{b_2}{a_2}{c_4}{m}\!\!
\makeSymbol{\includegraphics[scale=0.8]{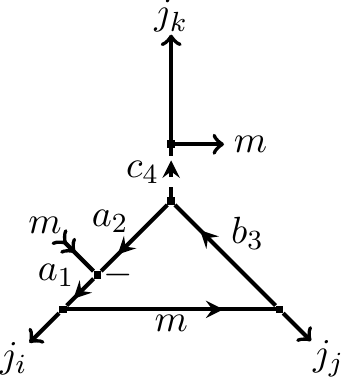}} 
\eeq
for the first term and 
\beq
\label{eqn:second}
\tsum\cdots \tsum_{a_3,b_3}(-)^{2(j_i+b_3)} \!\!\!\!\!
\makeSymbol{\includegraphics[scale=0.65]{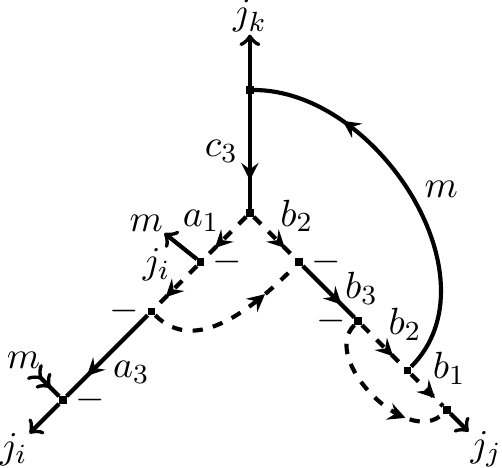}}
=
\tsum\cdots \tsum_{a_3,b_3}(-)^{2(j_i+b_3)}(-)^{b_1+b_3} 
\tinysixj{b_1}{b_3}{b_2}{j_j}{m}{m}\!\!\!\!
 \makeSymbol{\includegraphics[scale=0.8]{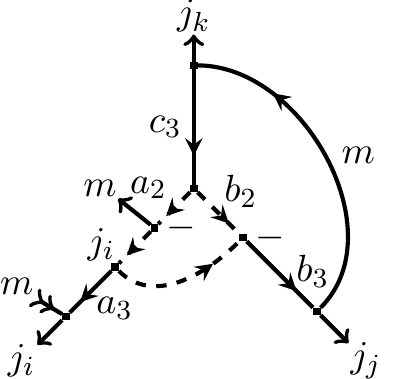}}
\eeq
for the second term. Since the newly created nodes are coplanar, $\hat{V}$ as well as $\hat{K}$ are vanishing on them no matter if they are gauge invariant or not. Therefore it suffices to calculate $\Tr[h_{s_i}\hat{K}\cdots$ on the inner parts. E.g. instead of considering the full diagram \eqref{eqn:first} it suffices to evaluate $\Tr[h_{s_i}\hat{K}\cdots$ on
\[
\makeSymbol{\includegraphics[scale=0.85]{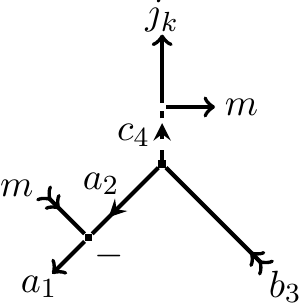}}~.
\]
Inserting \eqref{eqn:Kmatrix} and contracting with the last holonomy $h_{s_i}$ results in 
\begin{gather*}
\tsum_{\stackrel{a_3,a_4}{b_4,c_5}} (-)^{2a_4}K^{a_2 b_3;c_4}_{a_3 b_4;c_5}(j_k;m) \!\!\!\!
\makeSymbol{\includegraphics[scale=0.6]{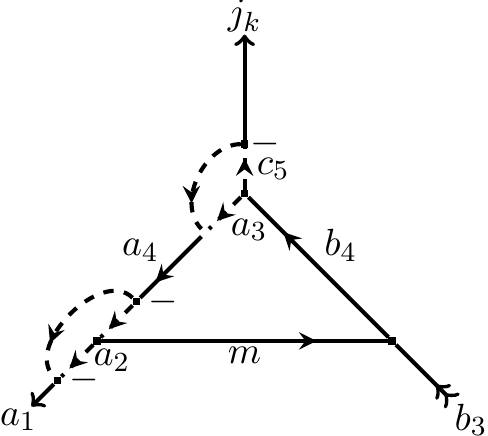}}\!\!
+\tsum_{\stackrel{a_3,a_4}{b_4,c_5}} (-)^{a_2+b_3+c_4+2m} \sixj{a_2}{j_k}{b_3}{m}{c_4}{a_3}
K^{b_3 c_4;a_3}_{b_4 c_5;a_4}(a_2;m) \!\!\!\!
\makeSymbol{\includegraphics[scale=0.65]{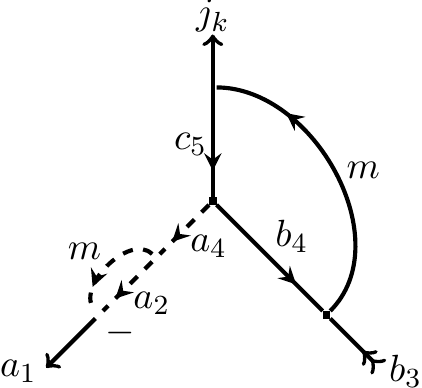}}\\
+\tsum_{\stackrel{a_3,a_4}{b_4,b_5,c_5}} (-)^{2(m+a_3)}(-)^{a_2+b_4+j_k} \sixj{a_2}{m}{b_3}{j_k}{c_4}{b_4}
K^{j_k a_2;b_4}_{c_5 a_3;b_5}(b_3;m) 
\makeSymbol{\includegraphics[scale=0.6]{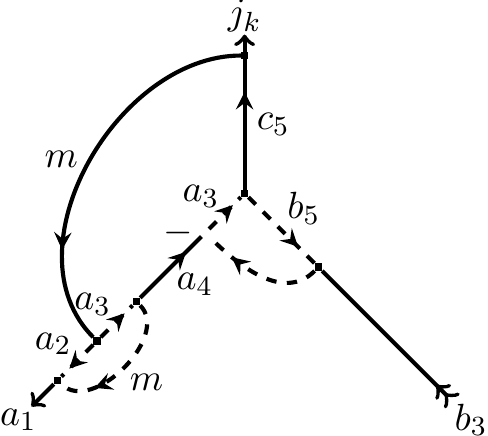}}~.
\end{gather*}
for this node. Finally, all dashed parts of the graphics can be erased due to \eqref{eqn:6jmove} and \eqref{eqn:firstorthogonality}. The other summands, \eqref{eqn:third} and \eqref{eqn:second} can be treated along the same lines. Just that in this case the inner parts are of the same type as the node shown on the right hand side of \eqref{eqn:A1}. Consequently, they have to be exerted again to remove the dashed line $m$ and move $\tip\,m$ to the right place before acting with $\hat{K}$. The final result of this computation is:
\begin{gather*}
\Tr_m[h_{s_i} \hat{K} h_{s_i^{-1}} h_{s_j} \hat{K} h_{s_j^{-1}} h_{s_k}V h_{s_k^{-1}}]
\makeSymbol{\includegraphics[scale=0.65]{Graphics/lorentz2-figure13}}
=
\tsum_{b_1,c_1}\sum_{c_2} V^{j_k}_{c_2}(j_i,j_j|c_1;m){\Big[}\\
\tsum_{\stackrel{a_1,a_2,b_2}{b_3,c_3,c_4}} (-)^{j_i-a_1+c_4-j_k}
\tinysixj{c_1}{c_3}{c_2}{j_k}{m}{m} \tinysixj{j_i}{m}{j_j}{c_3}{c_2}{b_1} K^{j_i b_1; c_3}_{a_1 b_2;c_4}(j_k;m)
\tinysixj{b_1}{b_3}{b_2}{j_j}{m}{m}\tinysixj{a_1}{b_3}{b_2}{a_2}{c_4}{m}
\Big[\\
\tsum_{\stackrel{a_3,a_4}{b_4,c_5}}  K^{a_2 b_3; c_4}_{a_3 b_4;c_5}(j_k;m)
 (-)^{a_2+a_3-a_4+b_4+c_5}
\tinysixj{a_2}{a_4}{a_3}{a_1}{m}{m}\tinysixj{a_3}{j_k}{b_4}{m}{c_5}{a_4}
\makeSymbol{\includegraphics[scale=0.65]{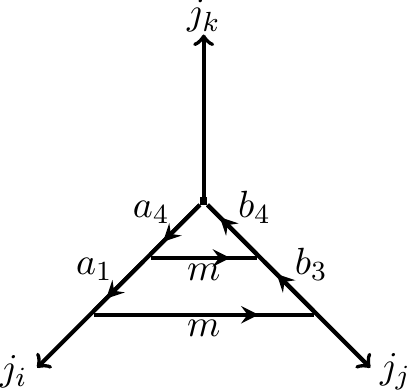}}
\\+
\tsum_{\stackrel{a_3,b_4}{c_5}}(-)^{b_3+c_4-a_2}\tinysixj{a_1}{j_k}{a_3}{c_4}{m}{b_3}
K^{b_3 j_k; a_3}_{b_4 c_5; a_1}(a_2;m) 
\makeSymbol{\includegraphics[scale=0.65]{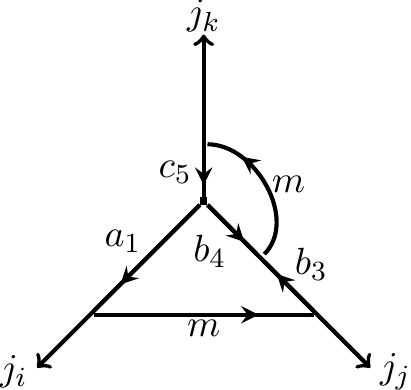}}
\\+
\tsum_{\stackrel{a_3,a_4}{b_4,b_5,c_5}} (-)^{a_1+a_2-a_3+a_4} (-)^{b_3+b_4+j_k+c_5}
\tinysixj{b_3}{j_k}{b_4}{c_4}{m}{a_2} K^{j_k a_2; b_4}_{c_5 a_3; b_5}(b_3;m) 
\tinysixj{a_2}{a_4}{a_3}{a_1}{m}{m}\tinysixj{a_3}{b_3}{b_5}{a_4}{c_5}{m}
\makeSymbol{\includegraphics[scale=0.65]{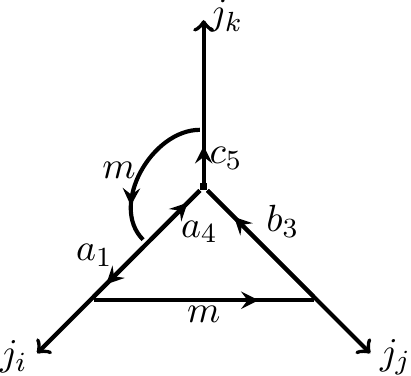}}
\Big]
\end{gather*}
\begin{gather*}
+
\tsum_{\stackrel{a_1,a_2}{b_2,b_3,c_3}} (-)^{a_1-j_i+b_1-b_3+c_1+c_2}
\tinysixj{c_1}{j_i}{c_2}{a_1}{m}{j_j} \tinysixj{a_1}{m}{j_j}{j_k}{c_1}{b_1} 
K^{b_1 j_k;a_1}_{b_2 c_3; a_2}(j_i;m) \tinysixj{b_1}{b_3}{b_2}{j_j}{m}{m}
\Big[\\
\tsum_{\stackrel{a_3,a_4}{b_4,c_4,c_5}}(-)^{a_3-a_4+a_5+b_4+c_5} 
\tinysixj{a_2}{c_4}{b_2}{m}{c_3}{j_i}\tinysixj{j_i}{b_3}{b_2}{a_2}{c_4}{m}
K^{a_3 b_3; c_3}_{a_4 b_4;c_4}(c_3; m) 
\tinysixj{a_3}{a_5}{a_4}{j_i}{m}{m} \tinysixj{a_4}{c_3}{b_4}{m}{c_5}{a_5}
\makeSymbol{\includegraphics[scale=0.65]{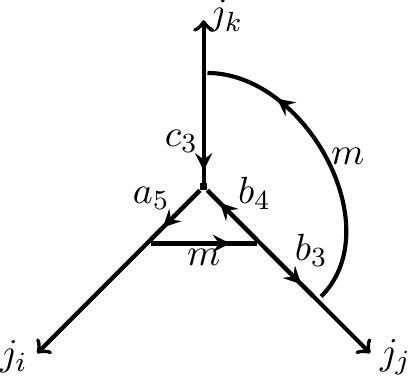}}
\\+
\tsum_{\stackrel{a_3,a_4}{b_4, c_4}} (-)^{j_i-a_2+a_4+b_2+c_3} 
\tinysixj{j_i}{a_4}{a_2}{a_3}{m}{m} \tinysixj{j_i}{b_3}{b_2}{a_4}{c_4}{m}
K^{b_3 c_3 ; a_4}_{b_4 c_4;a_1}(a_3;m)
\makeSymbol{\includegraphics[scale=0.65]{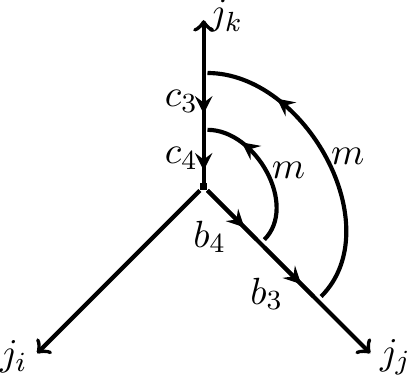}}
\\+
\tsum_{\stackrel{a_3,a_4}{b_4, b_5,c_4}}(-)^{a_2+a_4-a_5+b_2+c_4}
\tinyninej{b_2}{m}{b_3}{c_3}{a_3}{b_4}{a_2}{j_i}{m}
K^{c_3 a_3; b_4}_{c_4 a_4;b_5}(b_3;m)
\tinysixj{a_3}{a_5}{a_4}{j_i}{m}{m} \tinysixj{a_4}{b_3}{b_5}{a_5}{c_4}{m}
\makeSymbol{\includegraphics[scale=0.65]{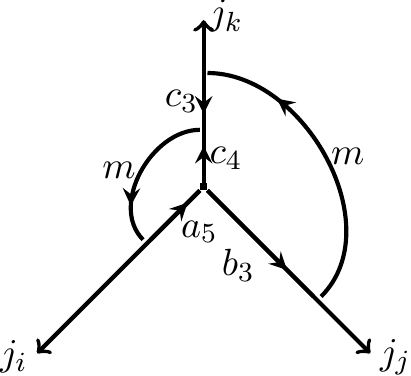}}
\Big]
\end{gather*}
\begin{gather*}
+\tsum_{\stackrel{a_1,a_2}{b_2,b_3,c_3}} (-)^{b_1-b_2+m}\tinyninej{j_i}{b_2}{j_k}{j_j}{b_1}{m}{c_2}{m}{c1}
K^{j_k j_i; b_2}_{c_3 a_1; b_3}(b_1;m)
\Big[\\
\tsum_{\stackrel{a_3,a_4}{b_4,c_4,c_5}} (-)^{a_2-a_3+a_4+j_j+b_3+b_4+c_3+c_4+c_5+m}
\tinysixj{a_1}{m}{b_3}{c_4}{c_3}{b_1} \tinysixj{a_1}{j_j}{b_1}{a_2}{c_4}{m}
K^{a_2 j_j; c_4}_{a_3 b_4; c_5}(c_3;m)
\tinysixj{a_2}{a_4}{a_3}{a_1}{m}{m} \tinysixj{a_3}{c_3}{b_4}{m}{c_4}{a_4}
\makeSymbol{\includegraphics[scale=0.65]{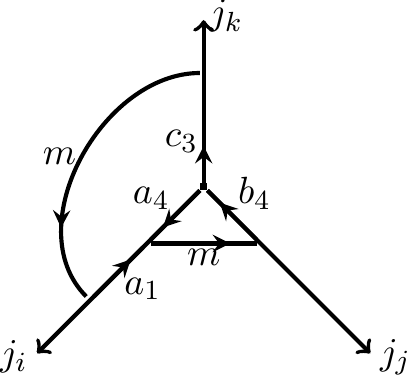}}
\\+
\tsum_{a_3,c_4} (-)^{a_2+b_3-c_3+m}\tinyninej{j_j}{b_1}{m}{c_3}{b_3}{a_1}{a_3}{m}{a_2}
K^{j_j c_3; a_3}_{b_4 c_4; a_1}(a_2;m) 
\makeSymbol{\includegraphics[scale=0.65]{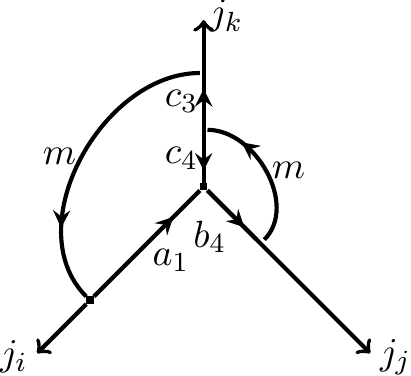}}
\\+
\tsum_{\stackrel{a_3,a_4}{b_4, b_5,c_4}}  (-)^{a_4-a_3+j_j-b_1+c_3+c_4+m} 
\tinysixj{b_1}{b_4}{b_3}{j_j}{m}{m}\tinysixj{a_1}{b_4}{b_3}{a_2}{c_3}{m}
K^{c_3 a_2; b_4}_{c_4 a_3; b_5}(j_j;m) 
\tinysixj{a_2}{a_4}{a_3}{a_1}{m}{m} \tinysixj{a_3}{j_j}{b_5}{a_4}{c_4}{m}
\makeSymbol{\includegraphics[scale=0.65]{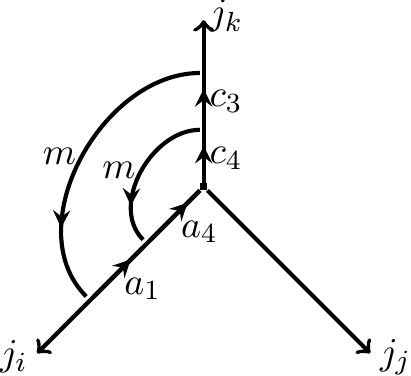}}
\Big]
\Big]
\end{gather*}

\subsection{Computation of $\Tr_m[h_{s_i} \hat{K} h_{s_i^{-1}} \hat{K} h_{s_k} V h_{s_k^{-1}}]$}
\label{ssec:trace2}
As before $ h_{s_k} \hat{V} h_{s_k^{-1}}$ produces a double non-invariant node transforming in the representation $\mathcal{H}^{\ast}_m\otimes\mathcal{H}_m$ where $^{\ast}$ denotes the adjoint. In contrast to the above the extrinsic curvature $\hat{K}$ is directly acting on this node. Therefore it is advisable to introduce an artificial coupling as it was done for the volume: 
\begin{gather*}
\makeSymbol{\includegraphics[scale=0.7]{Graphics/lorentz2-figure13}}
\xrightarrow{h_{s_k} \hat{V} h_{s_k^{-1}}}
\sum_{c_1, c_2} d_{c_1} (-)^{2 j_k}V^{j_k}_{c_2}(j_i,j_j|c_1; m)
\makeSymbol{\includegraphics[scale=0.7]{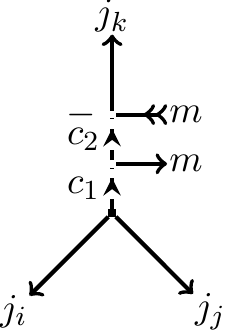}}\\
= \tsum_{c_1, m_1}\sum_{c_2} V^{j_k}_{c_2}(j_i,j_j|c_1;m)(-)^{c_1+c_2+m}
\sixj{m}{j_k}{m}{c_2}{m_1}{c_1} 
\makeSymbol{\includegraphics[scale=0.7]{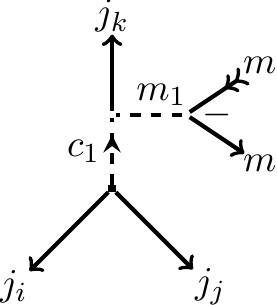}}
\end{gather*}
Recall that $m$ and $m_1$ are purely internal so that the curvature operator only registers a trivalent node transforming in spin $m_1$ and the previous results \eqref{eqn:Kmatrix} can be employed. Thus $h_{s_i^{-1}}\hat{K}$ is transforming the above expression into 
\begin{gather*}
\tsum_{c_1, m_1}\sum_{c_2} V^{j_k}_{c_2}(j_i,j_j|c_1;m)(-)^{c_1+c_2+m}
\sixj{m}{j_k}{m}{c_2}{m_1}{c_1}\left[ \tsum_{\stackrel{a_1,a_2}{b_1, c_3}} K^{j_i j_j; c_2}_{a_1b_1; c_3}(j_k;m_1) (-)^{2a_1} 
\makeSymbol{\includegraphics[scale=0.7]{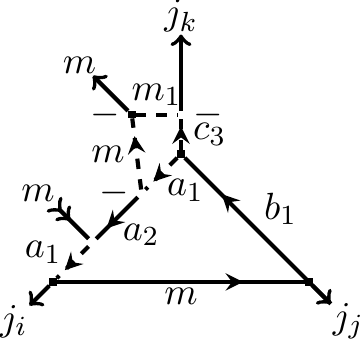}}
+\cdots\right] ~.
\end{gather*}
The link $m_1$ can be decoupled and parts of the internal lines can be removed: 
 \begin{gather*}
\makeSymbol{\includegraphics[scale=0.8]{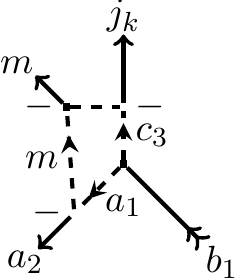}}
=\tsum_{c_3} (-)^{2c_4}
\makeSymbol{\includegraphics[scale=0.8]{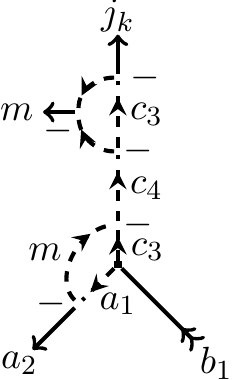}}\\
 =\tsum_{c_4} (-)^{2c_4} (-)^{b_1+j_k-a_1} \sixj{a_1}{c_4}{b_1}{m}{c_3}{a_2} 
 \sixj{m}{j_k}{m}{c_3}{m_1}{c_4}
\makeSymbol{\includegraphics[scale=0.8]{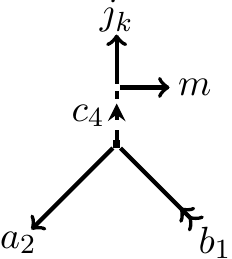}}
 \end{gather*}
The rest of the calculation is completely equivalent to the one in the previous section. With all other terms one can proceed similarly resulting in:
\begin{gather*}
\tsum_{c_1,m_1}\sum_{c_2} V^{j_k}_{c_2}(j_i,j_j|c_1; m) (-)^{c_1+c_2+m}
\tinysixj{m}{j_k}{m}{c_2}{m_1}{c_1}
\Big[
\tsum_{\stackrel{a_1,a_2}{b_1,c_3}} K^{j_i j_j; c_2}_{a_1 b_1; c_3}(j_k;m_1)
\Big[\\
\tsum_{\stackrel{a_3,a_4}{b_2,c_4,c_5}} (-)^{a_3+a_4+m} (-)^{b_2-b_1+j_k+c_5+m_1}
\tinysixj{a_1}{c_4}{b_1}{m}{c_3}{a_2} \tinysixj{j_k}{m}{c_3}{m}{m_1}{c_4} 
K^{a_2 b_1; c_4}_{a_3 b_2; c_5}(j_k;m)
\tinysixj{a_2}{a_4}{a_3}{a_1}{m}{m} \tinysixj{a_3}{j_k}{b_2}{m}{c_5}{a_4} 
\makeSymbol{\includegraphics[scale=0.65]{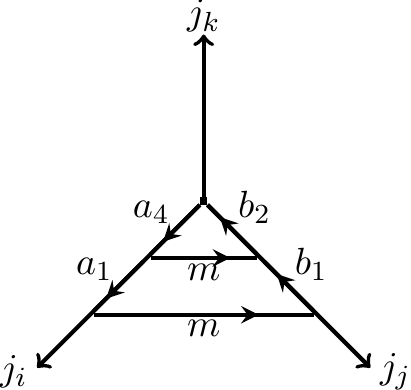}}
\\+
\tsum_{\stackrel{a_3,b_2}{c_4}} (-)^{a_1+a_2+a_3+b_1+c_3+m} 
\tinysixj{a_1}{m}{m_1}{a_2}{a_3}{m}\tinysixj{a_1}{j_k}{m_1}{b_1}{a_3}{c_3}
K^{b_1 j_k; a_3}_{b_2 c_4; a_1}(a_2;m)
\makeSymbol{\includegraphics[scale=0.65]{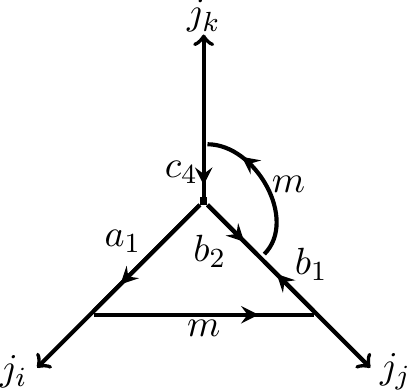}}
\\+
\tsum_{\stackrel{a_3,a_4}{b_2,b_3,c_4}}(-)^{a_2+a_3+a_4+b_1-j_k+c_3+c_4+m}
\tinyninej{a_1}{b_1}{c_3}{a_2}{b_2}{j_k}{m}{m}{m_1}
K^{j_k a_2; b_2}_{c_4 a_3; b_3}(b_1; m)
\tinysixj{a_2}{a_4}{a_3}{a_1}{m}{m}\tinysixj{a_3}{j_k}{b_2}{m}{c_5}{a_4}
\makeSymbol{\includegraphics[scale=0.65]{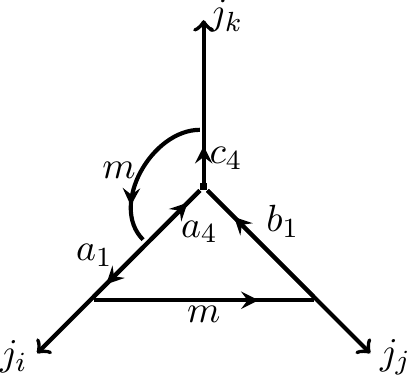}}
\Big]
\end{gather*}
\begin{gather*}
\tsum_{\stackrel{a_1,a_2,a_3}{b_1,c_3}} (-)^{a_2+a_3-j_i+j_j+c_2+m} 
\tinysixj{j_i}{j_k}{j_j}{m_1}{c_2}{a_1} 
K^{j_j j_k; a_1}_{b_1 c_3; a_2}(j_i;m_1)
\tinysixj{j_i}{m}{a_2}{m}{m_1}{a_3}\\
\Big[
\tsum_{\stackrel{a_4,a_5}{b_2, c_4,c_5}} (-)^{2m}(-)^{a_4+a_5+b_1+b_2+c_4+c_5} 
\tinysixj{a_2}{c_4}{b_1}{m}{c_3}{a_3} 
K^{a_3 b_1; c_4}_{a_4 b_2; c_5}(c_3;m)
\tinysixj{a_3}{a_5}{a_4}{j_i}{m}{m}\tinysixj{a_4}{c_3}{b_2}{m}{c_5}{a_5}
\makeSymbol{\includegraphics[scale=0.65]{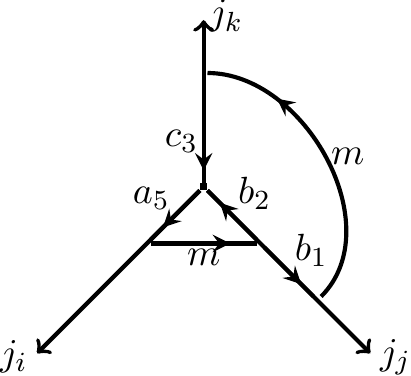}}
\\+
\tsum_{b_4,c_2} (-)^{2j_i}K^{b_1 c_3; a_2}_{b_2 c_4; j_i} (a_3; m)
\makeSymbol{\includegraphics[scale=0.65]{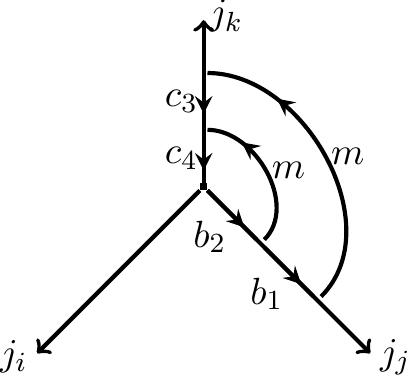}}
\\+
\tsum_{\stackrel{a_4,a_5}{b_2,b_3,c_4}} (-)^{j_i+a_2+a_4+a_5-c_3+c_4}
\tinysixj{a_2}{b_2}{b_1}{a_3}{c_3}{m}
K^{c_3 a_3: b_2}_{c_4 a_4;b_3}(a_2;m)
\tinysixj{a_3}{a_5}{a_4}{j_i}{m}{m} \tinysixj{a_4}{b_1}{b_3}{a_5}{c_4}{m}
\makeSymbol{\includegraphics[scale=0.65]{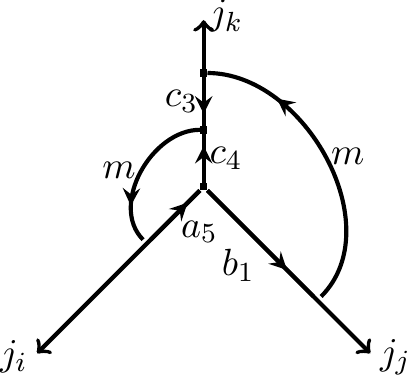}}
\Big]
\end{gather*}
\begin{gather*}
+ \tsum_{\stackrel{a_1,a_2}{b_1,b_2,c_3}} (-)^{j_i+j_k+b_1} 
\tinysixj{j_i}{m_1}{j_j}{j_k}{c_2}{b_1}
K^{j_k j_i;b_1}_{c_3 a_1;b_2}(j_j;m_1)\\
\Big[
\tsum_{\stackrel{a_3,a_4}{b_3,c_4,c_5}} (-)^{2 j_j} (-)^{m-a_2+a_3+a_4+b_3+c_3+c_4+c_5}
\tinyninej{a_1}{c_3}{b_2}{a_2}{c_4}{j_j}{m}{m}{m_1}
 K^{a_2 j_j;c_4}_{a_3 b_3; c_5}(c_3; m)
\tinysixj{a_2}{a_4}{a_3}{a_1}{m}{m}\tinysixj{a_3}{c_4}{b_3}{m}{c_5}{a_4}
\makeSymbol{\includegraphics[scale=0.65]{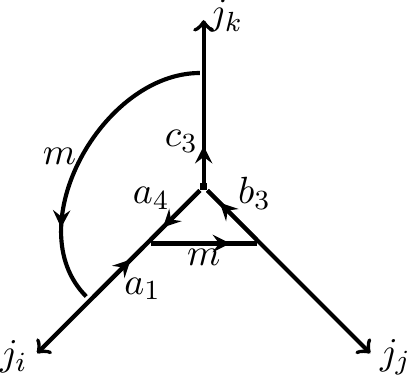}}
\\+
\tsum_{\stackrel{a_3,b_3}{c_4}} (-)^{j_j+c_3-a_2-m}
 \tinysixj{a_1}{m}{m_1}{a_2}{a_3}{m} \tinysixj{a_1}{j_j}{m_1}{c_3}{a_3}{b_2}
K^{j_j c_3; a_3}_{b_3 c_4; j_i}(a_2;m)
\makeSymbol{\includegraphics[scale=0.65]{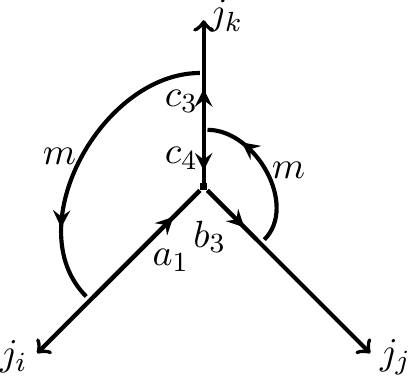}}
\\+
\tsum_{\stackrel{a_3,a_4}{b_3,b_4,c_4}} (-)^{a_3-a_4+b_2-b_3-c_3+c_4+m_1+m}
\tinysixj{a_1}{b_3}{b_2}{a_2}{c_3}{m}  \tinysixj{b_2}{m}{j_j}{m}{m_1}{b_3}
K^{c_3 a_2;b_3}_{c_4 a_3;b_4}(j_j;m)
\tinysixj{a_2}{a_4}{a_3}{a_1}{m}{m} \tinysixj{a_3}{j_j}{b_4}{a_4}{c_4}{m}
\makeSymbol{\includegraphics[scale=0.65]{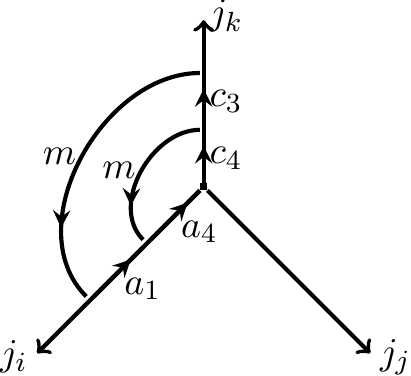}}
\Big]
\Big]
\end{gather*}

\subsection{Contraction with $\epsilon$}

To obtain the full action of $\hat{T}$ on a trivalent (invariant) node both trace contributions as computed in the previous sections must be summed up and contracted with the $\epsilon$-tensor. For the Euclidean constraint this antisymmetric contraction could be nicely absorbed in the loop trick \eqref{eqn:looptrick} and lead to major simplifications. Unfortunately, this does not happen for the remaining part of the scalar constraint. Since both, volume and extrinsic curvature, depend on whether one couples first holonomy $h_k$ or $h_i$ this contraction is not simplifying but complicating matters.\\
Yet, since we used an abstract calculus to evaluate the trace parts we are free to switch edges and nodes in the most advantageous position as long as the changes in (abstract) orientation and ordering are respected. For example:
\begin{gather*}
\Tr[h_{s_j}\hat{K}h_{s_j^{-1}}h_{s_i}\hat{K}h_{s_i^{-1}}h_{s_k}\hat{V}h_{s_k^{-1}}]\makeSymbol{\includegraphics[scale=0.7]{Graphics/lorentz2-figure13}}\!\!\!\!\!
=(-)^{j_i+j_j+j_k}\,
\Tr[h_{s_j}\hat{K}h_{s_j^{-1}}h_{s_i}\hat{K}h_{s_i^{-1}}h_{s_k}\hat{V}h_{s_k^{-1}}]\makeSymbol{\includegraphics[scale=0.7]{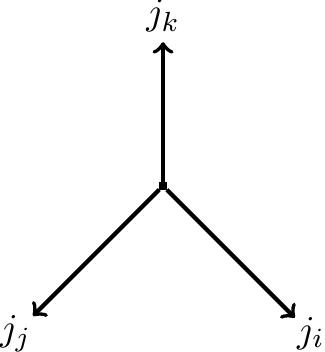}}
\end{gather*}
The trace can now be evaluated as above treating $i$ as $j$ and $j$ as $i$. Finally the edges should be flipped back:
\begin{gather*}
\Tr[h_{s_j}\hat{K}h_{s_j^{-1}}h_{s_i}\hat{K}h_{s_i^{-1}}h_{s_k}\hat{V}h_{s_k^{-1}}]\makeSymbol{\includegraphics[scale=0.7]{Graphics/lorentz2-figure13}}\!\!\!\!\!
=(-)^{j_i+j_j+j_k}
\tsum_{a_1,c_1}\sum_{c_2}
(-)^{j_k-c_2} V^{j_k}_{c_2}(j_i,j_j|c_1;m){\Big[}\\
\tsum_{\stackrel{b_1,b_2,a_2}{a_3,c_3,c_4}} (-)^{j_j-b_1+c_4-j_k}
\tinysixj{c_1}{c_3}{c_2}{j_k}{m}{m} \tinysixj{j_j}{m}{j_i}{c_3}{c_2}{a_1}(-)^{c_3-c_4} K^{a_1 j_j; c_3}_{a_2 b_1;c_4}(j_k;m)
\tinysixj{a_1}{a_3}{a_2}{j_i}{m}{m}\tinysixj{b_1}{a_3}{a_2}{b_2}{c_4}{m}
\Big[\\
\tsum_{\stackrel{b_3,b_4}{a_4,c_5}} (-)^{c_4-c_5} K^{a_3 b_2; c_4}_{a_4 b_3;c_5}(j_k;m)
 (-)^{b_2+b_3-b_4+a_4+c_5}
\tinysixj{b_2}{b_4}{b_3}{b_1}{m}{m}\tinysixj{b_3}{j_k}{a_4}{m}{c_5}{b_4}
(-)^{j_k+j_j+j_i}
\makeSymbol{\includegraphics[scale=0.65]{Graphics/lorentz3-figure12}}+\cdots\Big]+\cdots\Big]
\end{gather*}
where we used $V^{j_k}_{c_2}(j_i,j_j|c_1;m)=(-)^{j_k-c_2}V^{j_k}_{c_2}(j_j,j_i|c_1;m)$, $K^{a_1 j_j; c_3}_{a_2 b_1;c_4}=(-)^{c_3-c_4}$ and
\[
\makeSymbol{\includegraphics[scale=0.65]{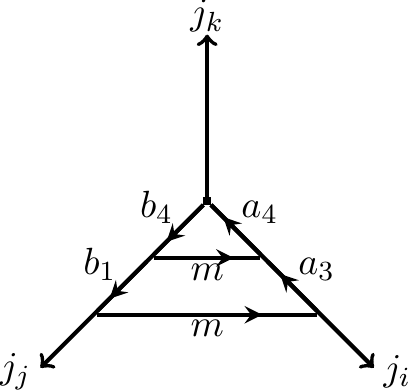}}
=
\underbrace{(-)^{j_k+a_4+b_4} (-)^{b_4+b_1+m} (-)^{j_j+b_1+m} (-)^{a_4+a_3+m} (-)^{j_i+a_3+m}}_{(-)^{j_k+j_j+j_i}}  \makeSymbol{\includegraphics[scale=0.65]{Graphics/lorentz3-figure12}}~.
\]
Note, that here the sign generated by the first switch of the edges is canceled by the one originating from restoring the old orientations. This is a generic property and applies to \emph{all} terms of the full expression. Only signs arising from volume and extrinsic curvature remain. The matrix elements corresponding to cyclic permutations of $(i,j,k)$ are simply obtained by exchanging the labels $(j_i,a_.)$ by $(j_j,b_.)$, $(j_j,b_.)$ by $(j_k,c_.)$ and so forth. \\
The second contribution $\Tr[h_{s_i}\hat{K}h_{s_i^{-1}}\hat{K}h_{s_k}\hat{V}h_{s_k^{-1}}]$ can be treated along the same lines. The value of $\Tr[h_{s_j}\hat{K}h_{s_j^{-1}}\hat{K}h_{s_k}\hat{V}h_{s_k^{-1}}]$ can be calculated by first flipping the edges so that $i$ can be treated as $j$ and vice versa and afterwards restoring the original orientation.

\section{Conclusion and Outlook}
\label{conclusion}

In this article we derived for the first time an explicit formula for the matrix elements of the full Hamiltonian constraint in LQG including the Lorentzian part. As already pointed out, this constraint plays a major role in any canonical quantization program for GR based on real Ashtekar-Barbero variables so that the methods developed in the course of the calculation are also of interest in these approaches, e.g. the master constraint approach. On the other hand, the tools developed to compute the action of the curvature (especially the loop trick \eqref{eqn:looptrick}) or extrinsic curvature can be easily adapted to models with non-graph changing operators, as the extended master constraint ansatz or AQG, by extending the loops involved in the regularization in such a way that no new links are created. \\
By exploiting several new recoupling identities, we significantly simplified the matrix element so that the recoupling part is totally captured in $6j$ and $9j$ symbols for which symmetry properties and explicit formulas are well known. Of course the final expression still depends on the volume but can be easily implemented on a computer for further investigations. We also expect to get interesting insight from a large $j$ expansion or the application in symmetry reduced models. Of special interest would be for example the recently introduced model \cite{iofrancesco1, iofrancesco2} which keeps the original $\SU(2)$ structure of the theory but has a diagonal volume operator so that it may be possible to give an analytical closed formula for the whole constraint within this Ansatz. Finally the presented analysis opens the way for a comparison with the covariant approach, because the spin foam vertex amplitudes are expected to be annihilated by the Hamiltonian constraint \cite{ioantonia}. As the matching between the canonical and covariant kinematics \cite{scattering3} led to the upgrade of the old Barret-Crane model \cite{BC} to the new EPRL-model \cite{EPRL}, the matching with the dynamical constraint is expected to shed new light onto the canonical-covariant joint theory.

{\acknowledgments
The authors wish to thank T.Thiemann for useful discussions and L.Cottrell for a careful reading of the manuscript.
The work of E.A. was partially supported by the grant of Polish Narodowe Centrum Nauki 
 nr DEC-2011/02/A/ST2/00300.}  
 A.Z. acknowledges financial support of the 'Elitenetzwerk Bayern' on the grounds of 'Bayerische Elitef\"order Gesetz'. 

\appendix
\section{More on $3j$'s, $6j$'s and $9j$'s}
\label{app:recoupling}
For self-containedness some important properties of $nj$-Symbols are listed here. Introductions to Recoupling theory can be found in various textbooks on quantum mechanics and quantum angular momentum, e.g. \cite{BrinkSatchler68}. For an extensive list of properties of $nj$-symbols see e.g. \cite{wolfram}

\begin{itemize}
\item[] {\bf 3j-Symbols}
	\begin{itemize}
		\item[] \emph{Relation to Clebsh-Gordan coefficients:}
			\[
			\scal{a,\alpha; b,\beta| c,\gamma}=(-)^{b-a+\gamma}\sqrt{2c+1}
			\threej{a}{\alpha}{b}{\beta}{c}{-\gamma}
			\]
			where $\ket{b,\beta;a,\alpha}= \ket{b,\beta}\otimes\ket{a,\alpha}$
			
		\item[]  \emph{Compatibility criteria}\\
			If one (or several) of the following rules is violated, 
			then $\threej{a}{\alpha}{b}{\beta}{c}{\gamma}$ is vanishing:
			\begin{itemize}
				\item $a,b,c\in\frac{1}{2}\N$, $a\pm\alpha\in\N$,  $-a\leq\alpha\leq a$, $\cdots$
				\item $\alpha+\beta+\gamma=0$
				\item $a+b+c\in\N$, $|a-b|\leq c\leq a+b$ (triangle inequality)
			\end{itemize}
		
		\item[]  \emph{Symmetries}
			\[
			\threej{a}{\alpha}{b}{\beta}{c}{\gamma}=(-)^{a+b+c}\threej{a}{-\alpha}{b}{-\beta}{c}{-\gamma}
			=(-)^{a+b+c}\threej{b}{\beta}{a}{\alpha}{c}{\gamma}=\threej{b}{\beta}{c}{\gamma}{a}{\alpha}
			\]
	\end{itemize}
\item[] {\bf 6j-Symbols}
	\begin{itemize}
		\item[]  \emph{Definition in terms of $3j$'s}
			\begin{align*}
			&\sixj{j_1}{j_4}{j_2}{j_5}{j_3}{j_6}=\\
			&\sum_{\mu_1,\cdots,\mu_6} (-)^{\sum\limits_{i=1}^6 (j_i-\mu_i)}
			\threej{j_1}{\mu_1}{j_2}{\mu_2}{j_3}{-\mu_3}\threej{j_1}{-\mu_1}{j_5}{\mu_5}{j_6}{\mu_6}
			\threej{j_4}{\mu_4}{j_5}{-\mu_5}{j_3}{\mu_3}\threej{j_4}{-\mu_4}{j_2}{-\mu_2}{j_6}{-\mu_6}
			\end{align*}
		\item[]  \emph{Symmetries}
			\begin{align*}
			\sixj{a}{d}{b}{e}{c}{f}=\sixj{b}{e}{a}{d}{c}{f}=\sixj{b}{e}{c}{f}{a}{d}=
			\sixj{d}{a}{e}{b}{c}{f}=\sixj{d}{a}{b}{e}{f}{c}=\sixj{a}{d}{e}{b}{f}{c}
			\end{align*}
		\item[]  \emph{Compatibility}
			\\[3pt]
			$\sixj{a}{d}{b}{e}{c}{f}=0$ unless the triangle inequalities hold for $\{a,b,c\}, \{a,e,f\},\{d,b,f\}$ 
			and $\{d,e,c\}$\\[3pt]
		\item[]  \emph{Orthogonality}
			\[
			\sum_x d_x \sixj{a}{d}{b}{e}{x}{c}\sixj{a}{d}{b}{e}{x}{c'}=\delta_{c,c'}\frac{1}{d_c}
			\]
			if the compatibility requirements are fulfilled.
	\end{itemize}
\item[]{\bf 9j-Symbols}
	\begin{itemize}
		\item[]  \emph{Definition by $3j$'s}
			\begin{gather*}
			\ninej{j_1}{j_2}{j_3}{j_4}{j_5}{j_6}{j_7}{j_8}{j_9}=
			\sum_{\mu_1,\dots,\mu_9} (-)^{\sum\limits_i^9(j_i-\mu_i)} 
			\threej{j_1}{\mu_1}{j_2}{\mu_2}{j_3}{\mu_3}\threej{j_4}{\mu_4}{j_5}{\mu_5}{j_6}{\mu_6}
			\threej{j_7}{\mu_7}{j_8}{\mu_8}{j_9}{\mu_9}\\\times
			\threej{j_1}{-\mu_1}{j_4}{-\mu_4}{j_7}{-\mu_7}\threej{j_2}{-\mu_2}{j_6}{-\mu_6}{j_8}{-\mu_8}
			\threej{j_3}{-\mu_3}{j_6}{-\mu_6}{j_9}{-\mu_9}
			\end{gather*}
		\item[]  \emph{Definition by $6j's$}
			\[
			\ninej{a}{f}{r}{d}{q}{e}{p}{c} {b}:=\sum_x d_x (-1)^{2x} \sixj{a}{c}{b}{d}{x}{p} 
			\sixj{c}{e}{d}{f}{x}{q} \sixj{e}{a}{f}{b}{x}{r}
			\]
		\item[] \emph{Symmetries}
			\begin{align*}
			\ninej{a}{f}{r}{d}{q}{e}{p}{c}{b}=(-)^S\ninej{d}{q}{e}{a}{f}{r}{p}{c}{b}
			=(-)^S\ninej{a}{f}{r}{p}{c}{b}{d}{q}{e}
			=(-)^S\ninej{f}{a}{r}{q}{d}{e}{c}{p}{b}=(-)^S\ninej{a}{r}{f}{d}{e}{q}{p}{b}{c}
			\end{align*}
			where $S=a+b+c+d+e+f+p+q+r$.
		\end{itemize}
\end{itemize}
	
\section{Volume}
\label{app:Volume}
For the sake of completeness the volume operator is briefly reviewed and the graphical framework for computing the action is discussed here, closely following \cite{DePietriRovelli96} and \cite{Gaul}. Thereby we restrict our attention to the  Ashtekar-Lewandowski volume \cite{AshtekarLewand98}, which was also analyzed in greater detail in \cite{Brunnemann:2004xi}.

\subsection{General properties} 

Let $T_s$ be a cylindric function on a spin network $s$ and let $\mathcal{V}(\Gamma)$ denote the set of  nodes of the underlying graph $\Gamma$. The volume operator $\hat{V}$ acts on $T_s$ by 
\begin{equation}
  \label{volume}
   {\hat{V}} \, T_s = \sum_{v \in \mathcal{V}(\Gamma)}
         {\hat{V}}_v \, T_s ~,
\end{equation}
where 
\begin{equation}
  \label{V_v}
   \hat{V}_v = l^3_p\, \sqrt{\left| \frac{i}{16\cdot3!} 
        \sum_{e_{I}\cap e_{J} \cap e_{K}=v} 
        \epsilon(e_{I},e_{J},e_{K})\,  {W}_{[IJK]} 
        \right|} ~.
\end{equation}
The sum extends over all triples $(e_I, e_J, e_K)$ of edges adjacent to the vertex $v$, $l_p$ denotes the Planck length and $\epsilon(e_{I},e_{J},e_{K})$ is determined by the orientation of the tangents $\dot{e}_I$ at $v$, i.e. $\epsilon(e_{I},e_{J},e_{K})=\mathrm{sgn}[\det(\dot{e}_{I},\dot{e}_{J},\dot{e}_{K})]$. The grasping operator 
\beq
\label{eqn:grasping}
W_{[IJK]}:=\epsilon_{ijk} X^i_I X^j_J X^k_K
\eeq
depends on the right invariant vector fields (of $\SU(2)$)
\beq
X^i_I:=-i\Tr\left[(h_{e_I}\sigma^i)^T\frac{\partial\;}{\partial h_{e_I}}\right]
\eeq
along the edge $e_I$ at $v$. Here, $e_I$ was chosen to be outgoing of $v$, $\sigma^i$ are the Pauli matrices and $T$ denotes the transpose. The matrix elements $[\sigma^i]^A_B$ with spinorial indices $A,B=0,1$ and vector indices $i=1,2,3$ are natural intertwiners of spin and vector representation (i.e. $(1,1/2,1/2)$). Thus, $X^i_I $ inserts an intertwiner $[\sigma^i]^A_B$ at an edge labeled by the fundamental representation. Since an irrep $l$ is the completely symmetrized tensor product of $2l$ fundamental representations this can be immediately generalized to edges labeled by $l$. Similarly $\epsilon$ is a vector invariant intertwining $(1,1,1)$. Following the spirit of the graphical calculus, $W_{[IJK]}$ can be visualized by
\beq
\label{eqn:grasp}
\sqrt{3!}\quad \makeSymbol{\raisebox{.25\height}{\includegraphics[scale=0.8]{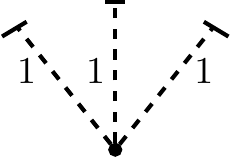}}}
\eeq
where each handle grasps an edge $e$ at $v$. Note, this grasping depends on the orientation of $e$. Yet, we want to use $3j$-symbols rather then the invariants $\sigma$ and $\epsilon$. These differ from the corresponding $3j$'s by a normalization constant. I.e. the $3j$'s are normalized to one while $\Tr[\epsilon^2]=3!$ and $\sum_i\Tr_l[(\sigma^i)^2]=4[l(l+1)(2l+1)]$ where $\Tr_l$ indicates that the trace is evaluated in spin $l$. Thus, $\sqrt{3!}$ in \eqref{eqn:grasp} stems from the normalization of $\epsilon$. Each grasp of an edge colored by $l$ gives additionally a factor $N_l:=-i2\sqrt{l(l+1)(2l+1)}$, in particular
\beq
\makeSymbol{\raisebox{.4\height}{\includegraphics{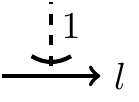}}}= N_l\cdot\;\makeSymbol{\raisebox{.1\height}{\includegraphics{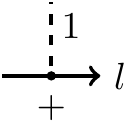}}}~.
\eeq
When using this formalism\footnote{This was developed first in \cite{DePietriRovelli96}. However, they used another calculus based on Temperley-Lieb algebras which yields different normalizations.}, one should keep in mind that the right invariant vector fields are derivative operators and hence only act on \emph{true} holonomies. On the other hand $W$ can not change the graph itself but only alter the intertwiners. Consequently, the links in \eqref{eqn:grasp} are only added \emph{at} the vertex (in a dashed environment) and can be erased again by pure recoupling theory.

\subsection{Trivalent nodes}
\label{app:V on trivalent}

The computational effort to determine the matrix elements of the volume operator is increasing heavily with the valency of the nodes. Therefore we only discuss the case of a trivalent node
\beq
\ket{v_x}:= \makeSymbol{\includegraphics[scale=0.8]{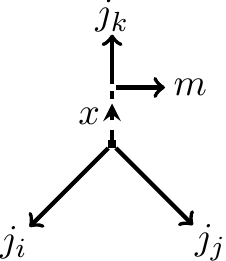}}
\eeq
transforming in spin $m$. On a trivalent node $\hat{V}_v$ reduces to $\frac{l_p^3}{4}\,\sqrt{\left| i  {W}_{ijk]} \right|}$ and the grasping yields
\beq
\label{4valente graspato}
\sqrt{3!}\quad\makeSymbol{\raisebox{.2\height}{\includegraphics[scale=0.8]{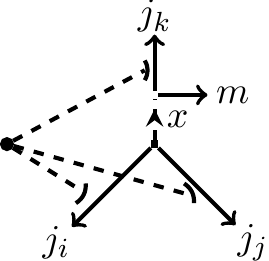}}}
=\sqrt{3!}N_{j_k} N_{j_i}N_{j_j}\sum_y d_y\;
\makeSymbol{\includegraphics[scale=0.8]{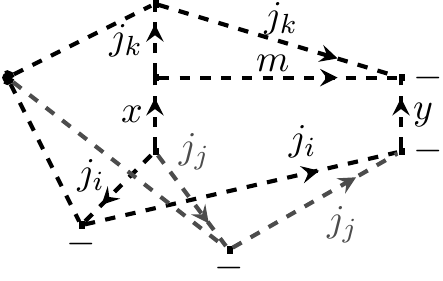}}\makeSymbol{\includegraphics[scale=0.8]{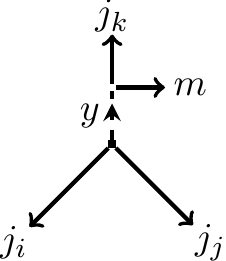}}~.
\eeq
In the second step a resolution of the identity in the intertwiner space was inserted (see section \ref{ssec:Imp_Id}). A careful evaluation by the usual methods reveals that 
\beq
\makeSymbol{\raisebox{-1.25\height}{\includegraphics[scale=0.8]{Graphics/lorentz3-figure44}}} = (-)^{x+j_k-m} \sixj{j_k}{x}{j_k}{y}{1}{m} \ninej{j_i}{1}{j_i}{x}{1}{y}{j_j}{1}{j_j}~.
\eeq
The $6j$ in this equation is only non-zero if $(x,y,1)$ obey the triangle inequality, i.e. if $|x-y|\leq1$. But for $x=y$ the $9j$-symbol is vanishing since it is anti-symmetric when switching first and last column. On the other hand, the $6j$ and the $9j$ are symmetric under the exchange of $x$ and $y$ for $y=x\pm1$. Thus, if we work with rescaled nodes $\ket{v_x}_N=\sqrt{d_x}\ket{v_x}$ then
\[
\hat{W}_{[IJK]}\ket{v_x}_N= \sum_y \tilde{W}^x\,_y \ket{v_y}_N
\]
yields an antisymmetric matrix $\tilde{W}$ which has only sub- and super-diagonal non-zero entries:
\beq
\tilde{W}^x\,_{y} =\delta_{y,x\pm1} \sqrt{3!} \sqrt{d_x d_{y}} N_{j_k} N_{j_i}N_{j_j} (-)^{x+j_k-m}\sixj{j_k}{x}{j_k}{y}{1}{m} \ninej{j_i}{1}{j_i}{x}{1}{y}{j_j}{1}{j_j}~.
\eeq
Fortunately, this matrix is diagonalizable so that the square root of $W$ has a well-defined meaning. Suppose $U$ is the (unitary) map that maps $\{\ket{v_x}_N\}$ to the eigenbasis of $\tilde{W}$ then the matrix elements of the volume (compare with \eqref{eqn:voume1}) are finally given by
\beq
V^x\,_y(j_i,j_j|j_k;m)=\frac{l_p^3}{4}\sqrt{\frac{d_y}{d_x}} [U^{-1}]^x\,_w  \sqrt{|i[\tilde{W}_D]^w\,_z} U^z\,_y\,\Lambda
\eeq
where $\tilde{W}_D=U\tilde{W} U^{-1}$ is the diagonalized matrix. For small $m$ the intertwiner space is low dimensional and this diagonalization does not cause much problems:
\begin{itemize}
\item[{\bf $m=0$:}] For $m=0$ the intertwiner space is one-dimensional and therefore $V$ annihilates gauge-invariant trivalent nodes.

\item[{\bf $m=\frac{1}{2}$:}] Here, the intertwiner space is 2-dimensional and it is not hard to check that $V$ is diagonal with matrix elements
\[
 V^x\,_y(j_i,j_j,j_k|\frac{1}{2})=\delta^x_y \frac{l_p^3}{4} \left[ |i\sqrt{3!} \sqrt{d_{j_k+\frac{1}{2}}\, d_{j_k-\frac{1}{2}}} \,N_{j_k} N_{j_i}N_{j_j} \tinysixj{j_k}{j_k+\frac{1}{2}}{j_k}{j_k-\frac{1}{2}}{1}{\frac{1}{2}} \tinyninej{j_i}{1}{j_i}{j_k+\frac{1}{2}}{1}{j_k-\frac{1}{2}}{j_j}{1}{j_j}|\right]^{-\frac{1}{2}}~.
 \]

\item[{\bf $m=1$:}] For $m=1$ $\tilde{W}$ is a $3\times3$-matrix but does not have full rank. Nevertheless, it is diagonalizable when applying first a similarity transformation $S$ (see \cite{Brunnemann:2004xi}):
\begin{gather*}
\tilde{W}= 
\begin{pmatrix}
0 & w_1 & 0\\
-w_1 & 0 & w_2\\
0& -w_2 & 0
\end{pmatrix}
\xrightarrow{ \;S\;}
\begin{pmatrix}
0 & \frac{1}{w_1} \lambda^2 & 0\\
-w_1 & 0 & 0\\
0&- \frac{w_2}{w_1} & 0
\end{pmatrix}
\xrightarrow{\;U\;} \lambda
\begin{pmatrix}
-i & 0 & 0\\
0 & i & 0\\
0& 0 & 0
\end{pmatrix}\\[5pt]
\xrightarrow{\quad}
V= \frac{l_p^3}{4} 
\begin{pmatrix}
\frac{|w_1|^2}{\lambda} & 0 &- \frac{|w_1\,w_2|^2}{\lambda} \sqrt{\frac{2j_k+3}{2j_k-1}}\\
0& \lambda & 0\\
\frac{|w_1\,w_2|^2}{\lambda} \sqrt{\frac{2j_k-1}{2j_k+3}}& \frac{|w_2|}{|w_1|}  &\frac{|w_2|^2}{\lambda}
\end{pmatrix}
\end{gather*}
where 
\[
w_1=(-)^{2(j_k-1)} \sqrt{3!} \sqrt{d_{j_k-1}\, d_{j_k}} \,N_{j_k} N_{j_i}N_{j_j} \tinysixj{j_k}{j_k-1}{j_k}{j_k}{1}{\frac{1}{2}} \tinyninej{j_i}{1}{j_i}{j_k-1}{1}{j_k}{j_j}{1}{j_j}
\]
\[
w_2=(-)^{2j_k-1} \sqrt{3!} \sqrt{d_{j_k+1}\, d_{j_k}} \,N_{j_k} N_{j_i}N_{j_j} \tinysixj{j_k}{j_k+1}{j_k}{j_k}{1}{\frac{1}{2}} \tinyninej{j_i}{1}{j_i}{j_k+1}{1}{j_k}{j_j}{1}{j_j}
\]
and $\lambda^2=|w_1|^2+|w_2|^2$.

\end{itemize}

\end{document}